\documentclass[12pt]{article}

\usepackage[T1]{fontenc}
\usepackage{lmodern}

\usepackage{setspace}
\setdisplayskipstretch{1}

\usepackage[margin=1.25in]{geometry}

\usepackage[dvipsnames]{xcolor}
\usepackage{pdfpages}
\usepackage{float}
\usepackage{multirow}
\usepackage{caption}
\usepackage{subcaption}
\usepackage{epigraph}
\usepackage{titlesec}

\titleformat{\section}
  {\normalfont\large\bfseries}{\thesection}{1em}{}
\titleformat{\subsection}
  {\normalfont\normalsize\bfseries}{\thesubsection}{1em}{}

\AtBeginDocument{%
  \everymath{\color{Violet}}%
}
\AtBeginEnvironment{displaymath}{\color{Violet}}
\AtBeginEnvironment{equation}{\color{Violet}}
\AtBeginEnvironment{equation*}{\color{Violet}}
\AtBeginEnvironment{align}{\color{Violet}}
\AtBeginEnvironment{align*}{\color{Violet}}

\usepackage{mathtools}
\usepackage{amssymb}
\usepackage{amsthm}
\usepackage{amsfonts}
\usepackage{aliascnt}
\usepackage{accents}
\usepackage{dutchcal}
\usepackage{bbm}

\usepackage{enumitem}
\usepackage{sgame}
\usepackage{tikz}
\usepackage{tikz-cd}
\usetikzlibrary{calc,shapes,arrows}
\usepackage{tcolorbox}
\usepackage{listings}
\usepackage[normalem]{ulem}

\usepackage[round]{natbib}
\newcommand{\E}{\mathbb{E}}
\newcommand{\R}{\mathbb{R}}

\newtheorem{theorem}{Theorem}[section]

\newaliascnt{proposition}{theorem}
\newtheorem{proposition}[proposition]{Proposition}
\aliascntresetthe{proposition}

\newaliascnt{lemma}{theorem}
\newtheorem{lemma}[lemma]{Lemma}
\aliascntresetthe{lemma}

\newaliascnt{corollary}{theorem}
\newtheorem{corollary}[corollary]{Corollary}
\aliascntresetthe{corollary}

\newaliascnt{claim}{theorem}
\newtheorem{claim}[claim]{Claim}
\aliascntresetthe{claim}

\theoremstyle{definition}

\newaliascnt{definition}{theorem}
\newtheorem{definition}[definition]{Definition}
\aliascntresetthe{definition}

\newaliascnt{example}{theorem}

\aliascntresetthe{example}

\newaliascnt{assumption}{theorem}

\aliascntresetthe{assumption}

\newaliascnt{condition}{theorem}

\aliascntresetthe{condition}

\newaliascnt{question}{theorem}

\aliascntresetthe{question}

\newaliascnt{remark}{theorem}

\aliascntresetthe{remark}

\newaliascnt{remarks}{theorem}

\aliascntresetthe{remarks}

\newaliascnt{aside}{theorem}

\aliascntresetthe{aside}

\newaliascnt{note}{theorem}

\aliascntresetthe{note}

\usepackage{thmtools}
\usepackage{thm-restate}

\usepackage{hyperref}

\hypersetup{
    colorlinks=true,
    linkcolor=OrangeRed,
    filecolor=Thistle,
    urlcolor=Thistle,
    citecolor=Thistle,
}

\usepackage[nameinlink]{cleveref}

\crefname{theorem}{theorem}{theorems}
\Crefname{theorem}{Theorem}{Theorems}
\crefname{proposition}{proposition}{propositions}
\Crefname{proposition}{Proposition}{Propositions}
\crefname{lemma}{lemma}{lemmas}
\Crefname{lemma}{Lemma}{Lemmas}
\crefname{corollary}{corollary}{corollaries}
\Crefname{corollary}{Corollary}{Corollaries}
\crefname{claim}{claim}{claims}
\Crefname{claim}{Claim}{Claims}
\crefname{definition}{definition}{definitions}
\Crefname{definition}{Definition}{Definitions}
\crefname{example}{example}{examples}
\Crefname{example}{Example}{Examples}
\crefname{assumption}{assumption}{assumptions}
\Crefname{assumption}{Assumption}{Assumptions}

\makeatletter
\let\cref@old@isrefconsecutive\cref@isrefconsecutive
\def\cref@isrefconsecutive#1#2{%
  \begingroup
    \def\cref@assumptiontype{assumption}%
    \cref@gettype{#1}{\cref@typea}%
    \ifx\cref@typea\cref@assumptiontype
      \endgroup
      \@cref@refconsecutivefalse
    \else
      \endgroup
      \cref@old@isrefconsecutive{#1}{#2}%
    \fi
}
\makeatother

\crefname{condition}{condition}{conditions}
\Crefname{condition}{Condition}{Conditions}
\crefname{question}{question}{questions}
\Crefname{question}{Question}{Questions}
\crefname{remark}{remark}{remarks}
\Crefname{remark}{Remark}{Remarks}
\crefname{remarks}{remarks}{remarks}
\Crefname{remarks}{Remarks}{Remarks}
\crefname{aside}{aside}{asides}
\Crefname{aside}{Aside}{Asides}
\crefname{note}{note}{notes}
\Crefname{note}{Note}{Notes}
\crefname{appendix}{appendix}{appendices}
\Crefname{appendix}{Appendix}{Appendices}

\newcommand{\secref}[1]{\hyperref[#1]{\S\ref*{#1}}}

\newcommand{\Conditionref}[1]{\hyperlink{eqcond:#1}{Condition~(#1)}}

\definecolor{backcolour}{rgb}{0.63,0.79,0.95}
\lstdefinestyle{mystyle}{
  backgroundcolor=\color{backcolour},
  basicstyle=\ttfamily\footnotesize,
  breakatwhitespace=false,
  breaklines=true,
  captionpos=b,
  keepspaces=true,
  numbers=left,
  numbersep=5pt,
  showspaces=false,
  showstringspaces=false,
  showtabs=false,
  tabsize=2
}
\begin{document}
\title{More Goods Means More Commitment}
\author{Mark Whitmeyer\footnote{Arizona State University. \href{mailto:mark.whitmeyer@gmail.com}{mark.whitmeyer@gmail.com}. I solicited feedback from \href{refine.ink}{refine.ink} and used ChatGPT throughout the research process. For KS.}}
\date{\today}
\maketitle

A seller with many independently valued goods can predict a buyer's total willingness to pay more accurately than her value for any individual good. Scope abrogates uncertainty and the buyer's rents.

A seller of a durable good who cannot commit to future prices competes with her future selves driving prices downward. Patience and lack of commitment extinguish the monopoly power of the seller and her profits.

Set these forces opposite. Which dominates? Will \textbf{scope} grant \textbf{commitment}? 

To wit, a seller without commitment to future offers interacts with a buyer over time. There are many goods for sale. At each date, the seller jointly announces a menu of subsets of the remaining inventory and current prices. The buyer may purchase one of the offered subsets or wait. In the next period, the seller may again offer any menu of the remaining goods. Neither can commit to future behavior.

The buyer's values are i.i.d. draws from a binary distribution supported on $\{\ell,b\}$ ($0<\ell<b$, \(q \coloneqq \mathbb{P}(b)\)). In the first part of \Cref{thm:main}, I show that in any equilibrium, 
as players become maximally patient ($\delta_K\uparrow1$) and scope increases sufficiently fast compared to patience ($K(1-\delta_K)\to\infty$, with \(K\) the number of goods), revenue per good converges to its expected value \(\mu=\ell+q(b-\ell)\) and buyer utility per good converges to zero. Scope can overpower the lack of commitment.


But scope need not. Sufficiently frequent revision opportunities (vis-a-vis scope) flips the extraction result. In the second part of \Cref{thm:main}, I show that when $(1-\delta_K)\mathrm{e}^{AK}\to0$, for a finite constant $A$ that depends only on $(\ell,b,q)$, seller revenue per good approaches $\ell$ and buyer utility per good approaches $\mu-\ell$ in every equilibrium. 

\paragraph{Intuition.} The forces behind the two directions of the theorem are different. When \(K(1-\delta_K)\to\infty\), the buyer's value per good for the grand bundle becomes nearly predictable, and scope grows quickly relative to the buyer's value from waiting. If equilibrium revenue per good were appreciably below the resulting bundle-value benchmark, the seller could profitably deviate only at the current date by offering the grand bundle at a price above her equilibrium revenue but sufficiently discounted that rejection is rare, returning to the equilibrium strategy after any rejection. 

In contrast, when \((1-\delta_K)\mathrm e^{AK}\to0\), delay is sufficiently cheap that the all-high type can imitate the all-low type at little cost. This ultimately bounds the seller's excess revenue in a recursive manner. If a purchase leaves sufficient probability that all the previously remaining values were low, the current bound depends on a continuation problem with fewer goods. If rejection raises that probability enough, the current bound depends on a continuation problem with the same goods but a greater probability that all are low. Otherwise, every purchase would lower that probability too far for the fewer-goods continuation bound to apply, and Bayes' rule bounds the probability that such a purchase occurs before one of the preceding two cases is reached. Each step introduces a loss proportional to \(1-\delta_K\), while iterating the argument across inventories and beliefs amplifies this loss by at most an exponential factor in \(K\). 

\section{The Game}\label{sec:model}

There are \(K \geq 1\) goods, with \([K]\coloneqq\{1,\ldots,K\}\). There is a buyer (he), who privately observes his vector of valuations for the goods \(v=(v_1,\ldots,v_K)\in V\coloneqq\{\ell,b\}^K\), drawn from a common prior \(P\in\Delta(V)\), where \(0<\ell<b\). I assume that the components of \(v\) are i.i.d., with \(P(v_i=b)=q\in(0,1)\) and write \(\mu=\ell+q(b-\ell)\). For \(S\subseteq[K]\), write \(v(S)\coloneqq\sum_{i\in S}v_i\).

There is also a seller (she). She and the buyer interact over infinitely many periods. For each \textit{inventory} \(I\subseteq[K]\), let \(\mathcal M(I)\) denote the set of finite menus \(M\) satisfying
\[(\emptyset,0)\in M, \qquad \text{\&} \qquad M\setminus\{(\emptyset,0)\} \subseteq \{(S,p)\colon \emptyset\neq S\subseteq I, p\in\R\}.\]
I write \(0\) for the offer corresponding to rejection by the buyer, \((\emptyset,0)\). At date \(t\in\{0,1,\ldots\}\), denote the publicly observed remaining inventory by \(I_t\). If \(I_t\neq\emptyset\), the seller jointly announces a costless public message \(s_t\in[0,1]\) and a menu \(M_t\in\mathcal M(I_t)\).\footnote{I use the message as a payoff-irrelevant public label for the realization of the seller's mixed strategy and only in the equilibrium-existence proof (the two directions of the main theorem don't use it). In particular, I use a sequence of finite games as an approximation of the true game, with the complication that when I pass to the limit, distinct menus in the approximating games may converge to the same menu while carrying different buyer responses, posteriors, and continuation play. To sidestep this issue, I assign different messages whenever the limiting menus coincide but the associated buyer responses, posteriors, or continuation play differ. Note that an exogenous public sunspot uniformly distributed on \([0,1]\) and observed before the seller announces her menu would also do the trick.} The buyer observes both message and menu before responding, selecting one offer \(a_t=(S_t,p_t)\in M_t\). Following his selection, the game moves to the next period, with \(I_{t+1}=I_t\setminus S_t\). The game ends when \(I_t=\emptyset\). Costs are zero. The seller and buyer share a common discount factor \(\delta\in(0,1)\), and their discounted payoffs are
\begin{equation}\label{eq:payoffs}R=\E\left[\sum_{t\geq0}\delta^t p_t\right] \qquad \text{\&} \qquad U=\E\left[\sum_{t\geq0}\delta^t\left(v(S_t)-p_t\right)\right].\end{equation}

The public history at the beginning of date \(t\) is \(h_t\coloneqq\left((M_\tau,s_\tau,a_\tau)\right)_{\tau<t}\), where \(M_\tau\in\mathcal M(I_\tau)\), \(a_\tau=(S_\tau,p_\tau)\in M_\tau\), and \(I_{\tau+1}=I_\tau\setminus S_\tau\) for every \(\tau<t\), starting from \(I_0=[K]\). I term \(h_t\) a \textit{seller history}. A seller history \(h_t\) is \textit{terminal} if \(I_t=\emptyset\) and \textit{nonterminal} otherwise. The game ends at a terminal seller history, and the seller moves only at a nonterminal seller history. After the seller chooses \(M_t\in\mathcal M(I_t)\) and \(s_t\in[0,1]\) at a nonterminal seller history, the public history when the buyer moves is \((h_t,M_t,s_t)\). The buyer's private history when he moves is \((v,h_t,M_t,s_t)\). At a seller history \(h\), write \(I(h)\) for the remaining inventory, \(P_h\) for her posterior on \( V\), and \(m(h)\coloneqq|I(h)|\).

A behavioral strategy for the seller, \(\sigma\), assigns to each nonterminal seller history \(h\) a probability distribution \(\sigma(\cdot\mid h)\in\Delta\left(\mathcal M(I(h))\times[0,1]\right)\). For each \(v\in V\), a behavioral strategy for the buyer, \(\beta_v\), assigns to each public history \((h,M,s)\) at which the buyer moves a probability distribution \(\beta_v(\cdot\mid h,M,s)\in\Delta(M)\). Note that both strategies may depend on the entire public history, while the buyer's strategy may also depend on his privately observed type.

An \textit{assessment} consists of a seller strategy \(\sigma\), buyer strategies \((\beta_v)_{v\in V}\), and a posterior over \( V\) at every seller history, including seller histories that occur with probability zero. After the buyer responds with \(a\in M\) at \((h,M,s)\), write the posterior at the ensuing seller history as \(P_{h,M,s,a}\).

Starting from a seller history \(h\), the assessment determines a distribution over all subsequent menus, messages, and responses. Let \(R(h)\) denote the seller's expected discounted revenue from that date onward and, conditional on the buyer having type \(v\), let \(u_v(h)\) denote his expected discounted utility. Starting from a public history \((h,M,s)\) at which the buyer moves, define \(R(h,M,s)\) and \(u_v(h,M,s)\) analogously. These expected payoffs are the players' \textit{continuation payoffs}. At a terminal seller history, \(R(h)=u_v(h)=0\). An assessment is \textit{finite} if all of these expected payoffs are real numbers.

For every nonterminal seller history \(h\), every \(M\in\mathcal M(I(h))\), every \(s\in[0,1]\), and every \(a\in M\), define
\[\pi_a(h,M,s) \coloneqq \sum_{\hat v\in V} P_h(\hat v)\beta_{\hat v}(a\mid h,M,s).\]
With \(h\), \(M\), and \(s\) fixed, write this probability as \(\pi_a\).

\begin{definition}\label{def:equilibrium}
An \textit{equilibrium} is a finite assessment that satisfies:
\begin{enumerate}[nosep,noitemsep]
\item[(i)]\hypertarget{eqcond:ii}{} At every nonterminal seller history \(h\), no seller strategy delivers the seller expected discounted revenue greater than \(R(h)\) when play starts at \(h\), holding fixed the buyer strategies and the posteriors specified by the assessment.

\item[(ii)]\hypertarget{eqcond:iii}{} At every public history \((h,M,s)\) at which the buyer moves and for every \(v\in V\), no buyer strategy for type \(v\) delivers him expected discounted utility greater than \(u_v(h,M,s)\) when his type is \(v\) and play starts at \((h,M,s)\), holding fixed the seller strategy and the posteriors specified by the assessment.

\item[(iii)]\hypertarget{eqcond:iv}{} The posterior at the initial seller history is \(P\). At every nonterminal seller history \(h\), for every \(M\in\mathcal M(I(h))\), every \(s\in[0,1]\), and every \(a\in M\), the posterior at the seller history following the buyer's response satisfies
\begin{equation}\label{eq:bayes} P_h(v)\beta_v(a\mid h,M,s) = \pi_a P_{h,M,s,a}(v) \qquad\text{for every }v\in V. \end{equation}
\end{enumerate}
\end{definition}
\Conditionref{iv} imposes no signaling what you do not know and Bayes' rule whenever possible.\footnote{And so this solution concept is a standard refinement of weak Perfect Bayesian Equilibrium.} Because the seller does not observe the buyer's type, her choice of menu and message leaves the type distribution equal to \(P_h\), including after an announcement that occurs with probability zero. The buyer's response may depend on his type, so Bayes' rule determines the posterior at the ensuing seller history whenever \(\pi_a>0\), while the posterior is unrestricted when \(\pi_a=0\).


For \(x\in\R\), write \(x^+\coloneqq\max\{x,0\}\). Fix a seller history \(h\) and two types \(v,\hat v\). Type \(\hat v\) can follow the buyer strategy assigned to \(v\) after every menu and message that follows \(h\). This produces the same distribution over trades and payments as when type \(v\) follows that strategy. Accordingly, \(v\)'s discounted value for the goods purchased minus \(\hat v\)'s discounted value for those goods is at most \(\sum_{i\in I(h)}(v_i-\hat v_i)^+\). And so imitation provides type \(\hat v\) expected utility at least \(u_v(h)-\sum_{i\in I(h)}(v_i-\hat v_i)^+\). Applying \Conditionref{iii} after each menu and message that the seller may choose at \(h\), and averaging over her choice, yields
\begin{equation}\label{eq:imitation-general} u_v(h)-u_{\hat v}(h) \leq \sum_{i\in I(h)}(v_i-\hat v_i)^+. \end{equation}

\begin{proposition}\label{prop:existence}
For every finite $K$, every prior $P\in\Delta(V)$, and every $\delta\in(0,1)$, an equilibrium exists.
\end{proposition}

\section{Scope vs. patience}\label{sec:main}
For a fixed prior, write $R_{K,\delta}$ and $U_{K,\delta}$ for an arbitrary equilibrium's payoffs. Recall \(\mu=\ell+q(b-\ell)\).
\begin{theorem}\label{thm:main}
Fix $0<\ell<b$ and $q\in(0,1)$. There is a finite constant $A>0$, depending only on $(\ell,b,q)$, such that, for any sequence $\delta_K\in(0,1)$ and any sequence of equilibria,
\begin{align}
K(1-\delta_K)\rightarrow\infty &\quad\Longrightarrow\quad \left(\frac{R_{K,\delta_K}}K,\frac{U_{K,\delta_K}}K\right) \rightarrow(\mu,0),\label{eq:extraction-regime} \qquad \text{\&}\\
(1-\delta_K)\mathrm{e}^{AK}\rightarrow 0 &\quad\Longrightarrow\quad \left(\frac{R_{K,\delta_K}}K,\frac{U_{K,\delta_K}}K\right) \rightarrow(\ell,\mu-\ell).\label{eq:revision-regime}
\end{align}
In either regime, expected discounted total surplus per good converges to $\mu$.
\end{theorem}

I prove the two limits in \Cref{prop:extraction-bound,prop:revision-bound}.\footnote{Note that there is an interpretation of the theorem in terms of increasingly frequent offers: let consecutive offers be $\Delta_K$ units apart, and let the common discount rate be $\rho>0$. Then $\delta_K=\mathrm{e}^{-\rho\Delta_K}$. When $\Delta_K\to0$, the condition in \eqref{eq:extraction-regime} under which the seller extracts the full surplus becomes $K\Delta_K\to\infty$, and the condition in \eqref{eq:revision-regime} under which her revenue per good converges to $\ell$ becomes $\Delta_K \mathrm{e}^{AK}\to0$.} A seller who can commit to a single grand-bundle offer can obtain revenue per good approaching \(\mu\). For each \(\varepsilon>0\), suppose that she prices the bundle at \(K(\mu-\varepsilon)\). By the law of large numbers, the buyer's total value exceeds this price with probability approaching one as \(K\to\infty\), so expected revenue per good approaches \(\mu-\varepsilon\). Letting \(\varepsilon\downarrow0\) delivers the benchmark \(\mu\). And so \eqref{eq:extraction-regime} provides the seller the same limiting revenue per good even though she cannot commit. In contrast, under \eqref{eq:revision-regime}, seller revenue per good converges to \(\ell\) and buyer utility per good converges to \(\mu-\ell\). In either regime, expected discounted total surplus per good converges to \(\mu\), i.e., any waste from delay or failure to sell vanishes per good.

\begin{corollary}\label{cor:limits}
\(\lim_{\delta\uparrow1}\lim_{K\to\infty}\frac{R_{K,\delta}}K=\mu\) and \(\lim_{K\to\infty}\lim_{\delta\uparrow1}\frac{R_{K,\delta}}K=\ell\) in any sequence of equilibria.
\end{corollary}
\begin{proof}
For fixed $\delta<1$, the upper bound on $\mu-R_{K,\delta}/K$ in \eqref{eq:main-extraction} converges to zero as $K\to\infty$.\footnote{Forgive the forward reference.} For fixed $K$, the upper bound on $R_{K,\delta}/K-\ell$ in \eqref{eq:main-revision} converges to zero as $\delta\uparrow1$.
\end{proof}

\subsection{Joint vs. independent negotiations}

The main game is a joint negotiation. The seller has \(K\) goods, can offer arbitrary subsets of the remaining goods in each period, and her subsequent offers may depend on the entire history of purchases, rejections, menus and messages. This ability to link is essential for the main result. Let us now investigate the case in which such linkage is forbidden; \textit{viz.}, where the parties take part in \(K\) separate single-good games. In the negotiation for good \(i\), the seller may respond to what has happened in that negotiation--i.e., the seller's menus, messages, and beliefs depend only on that negotiation's history--but not to what has occurred in the negotiations for other goods. Likewise, the buyer's strategy depends only on the value of that good
and the history of that negotiation.

So, we are now studying $K$ independent copies of the one-good sales game (with the same value distribution and discount factor $\delta$). An equilibrium of this benchmark is a collection of one-good equilibria. What happens now under condition for seller extraction in joint negotiations (\(K(1-\delta_K)\rightarrow\infty\)) plus maximal patience \(\delta_K\uparrow1\)? \(R^{\mathrm{ind}}_{K,\delta_K}\) is the sum of the seller's revenues in a one-good equilibrium.
\begin{corollary}\label{cor:independent-negotiations}
Suppose $\delta_K\uparrow1$ and $K(1-\delta_K)\to\infty$. For any equilibrium of each joint sales game and any collection of one-good equilibria,
\[\frac{R^{\mathrm{ind}}_{K,\delta_K}}K\rightarrow\ell, \qquad \text{\&} \qquad \frac{R_{K,\delta_K}-R^{\mathrm{ind}}_{K,\delta_K}}K \rightarrow\mu-\ell=q(b-\ell)>0.\]
Buyer utility per good converges to zero under joint contracting and to $\mu-\ell$ under independent negotiations. Expected discounted total surplus per good converges to $\mu$ in both cases.
\end{corollary}
\begin{proof}
\Cref{thm:main} establishes the joint-contracting limits. Apply the estimate \eqref{eq:main-revision} with one good to each independent negotiation and average. The constants depend only on $(\ell,b,q)$, so
\[0\leq\frac{R^{\mathrm{ind}}_{K,\delta_K}}K-\ell \leq C\mathrm{e}^A(1-\delta_K)\rightarrow 0.\]
The estimate of the buyer's utility in \eqref{eq:main-revision}, also applied with one good, implies that average buyer utility converges to $\mu-\ell$. Combine with the output of \Cref{thm:main} to get the other values.
\end{proof}

Of course, there is nothing stopping the two players from following separate-sales strategies in the joint-negotiation game. But this won't be an equilibrium for sufficiently large \(K\). In particular, the seller has a profitable deviation by offering the grand bundle (all \(K\) goods together) today and, after rejection, returning next period to the prescribed separate negotiations.

Under separate sales, seller revenue per good converges to \(\ell\), while buyer utility per good converges to \(\mu-\ell\). The buyer's expected willingness to pay per good for the bundle rather than
wait is approximately
\[\mu-\delta_K(\mu-\ell) = \ell+(1-\delta_K)(\mu-\ell).\]
Under separate sales, the buyer captures all but \(O(1-\delta_K)\) of the difference in values \(b-\ell\), so his willingness to pay for immediate delivery differs by only \(O(1-\delta_K)\) between a good worth \(b\) and one worth \(\ell\). Because values are independent, the standard deviation of his average willingness to pay for the bundle is only \(O((1-\delta_K)/\sqrt K)\). Since its expectation exceeds the seller's separate-sales revenue per good by an amount of order \(1-\delta_K\), the seller can charge more than under separate sales and have the bundle accepted with probability approaching one. In the proof of \Cref{prop:exponential}, I use the same grand-bundle offer to derive a revenue bound for arbitrary equilibria.

\section{Full extraction}\label{sec:extraction}

In this section, I prove \eqref{eq:extraction-regime}, the full-extraction part of \Cref{thm:main}. I first derive preliminary bounds on payoffs in \eqref{eq:uniform-values}. I then prove in \Cref{lem:baseline} that the seller's continuation revenue is at least $\ell$ per remaining good in every equilibrium, whatever her posterior. I then consider a deviation in which the seller offers all goods together. Using the buyer's incentive to reject and the effect of rejection on the posterior (\Cref{lem:conditioning}), I derive a lower bound on equilibrium revenue that holds under any prior (\Cref{prop:exponential}). Finally, I prove that any random variable $X\in[0,b]$ satisfies the bound $\log\E[\mathrm e^{s(X-\E X)}]\leq s^2b^2/8$ for every $s\in\R$. 

Applying this bound to each valuation and using their independence allows me to bound a statistic for the buyer's value per good $w=v([K])/K$, $\log\E[e^{-\lambda w}]$, above by $-\lambda\mu+\lambda^2b^2/(8K)$ \eqref{eq:mgf}. As a result, the statistic $c_\lambda(P) \coloneqq -\lambda^{-1}\log\E[e^{-\lambda w}]$ falls below $\mu$ by at most $\lambda b^2/(8K)$. And the earlier bound in \Cref{prop:exponential} limits how far equilibrium revenue per good can fall below $c_\lambda(P)$. Combining these bounds yields \eqref{eq:main-extraction}, an upper bound on $\mu-R_{K,\delta}/K$, the difference between expected value per good and the seller's equilibrium revenue per good. This delivers \Cref{prop:extraction-bound} and taking $K(1-\delta_K)\to\infty$ then establishes its corollary \eqref{eq:extraction-regime}.

Fix a seller history \(h\), and write \(I=I(h)\), \(m=m(h)\), \(W=v(I)\), \(R=R(h)\), and \(u_v=u_v(h)\). The seller can choose the rejection-only menu forever, and the buyer can choose rejection forever. Hence, \(R\geq0\) and \(u_v\geq0\). Since both players have the same discount factor,
\begin{equation}\label{eq:surplus} R+\E_{P_h}[u_v] = \E_{P_h}\left[\sum_{t\geq0}\delta^t v(S_t)\right] \leq \E_{P_h}[W], \end{equation}
which implies \(\E_{P_h}[u_v]\leq mb\), so some \(\bar v\) with \(P_h(\bar v)>0\) satisfies \(u_{\bar v}\leq mb\). Applying \eqref{eq:imitation-general} with \(\hat v=\bar v\) delivers
\[u_v \leq u_{\bar v}+\sum_{i\in I}(v_i-\bar v_i)^+ \leq mb+m(b-\ell) = m(2b-\ell)\]
for every \(v\in V\), including types with \(P_h(v)=0\), so
\begin{equation}\label{eq:uniform-values}
0\leq R\leq mb, \qquad \text{\&} \qquad 0\leq u_v\leq m(2b-\ell) \quad\text{for every }v.
\end{equation}

\begin{lemma}\label{lem:baseline}
At every seller history \(h\), every equilibrium satisfies \(R(h)\geq m(h)\ell\), whatever the posterior \(P_h\).
\end{lemma}

The lemma reveals that in any equilibrium, at any seller history, regardless of the seller's posterior, the seller earns at least \(\ell\) per remaining good. The idea behind the proof is to bootstrap any weaker (uniform) revenue floor. To that end, let \(\alpha\) be the largest number in \([0,1]\) such that, under every prior and equilibrium, the seller's continuation revenue is at least \(\alpha m'\ell\) at every seller history with \(m'\) remaining goods.

To improve upon this lower bound, fix any history with \(m\) remaining goods and consider the seller's deviation that prices every subset at \(c\ell\) per good, where \(c=\delta\alpha+(1-\delta)-\varepsilon\). If the buyer leaves \(m'\) goods unsold, with total value \(W'\) and continuation utility \(u_v'\), then choosing the offer for all \(m\) goods would give him the \(m'\) remaining goods as well at an additional cost of \(c\ell m'\). As a result, his optimal behavior requires \(W'-c\ell m'\leq\delta u_v'\). But at the next seller history, the seller's expected revenue from the \(m'\) unsold goods is at least \(\alpha m'\ell\), so applying
\eqref{eq:surplus} after each response and averaging over responses implies \(\E[u_v']\leq\E[W']-\alpha\ell\E[m']\). 

Combining these inequalities yields \((1-\delta)\left(\E[W']-\ell\E[m']\right) +\varepsilon\ell\E[m']\leq0\), and as both terms on the left are nonnegative, \(m'=0\) almost surely. This deviation results in all \(m\) goods selling for \(cm\ell\), which implies \(\alpha\geq c\). Taking \(\varepsilon\downarrow0\) produces \(\alpha\geq\delta\alpha+(1-\delta)\), and hence \(\alpha=1\).

For \Cref{prop:exponential}, consider the game with initial inventory $[K]$ and an arbitrary prior $P\in\Delta(V)$. Write \(w(v)=v([K])/K\). For \(\lambda>0\), define
\begin{equation}\label{eq:exponential-statistic}
c_\lambda(P)\coloneqq-\frac1\lambda\log\E_P[\mathrm{e}^{-\lambda w}].
\end{equation}
Jensen's inequality implies $c_\lambda(P)\leq\E_P[w]$.

\begin{proposition}\label{prop:exponential}
Fix $K,\delta$, and $\lambda>0$ such that $t\coloneqq\lambda(1-\delta)b\geq1$. Every equilibrium under any prior $P$ satisfies \(\frac RK\geq c_\lambda(P)-b\frac{\log t+2}{t}\).
\end{proposition}

The role of \Cref{prop:exponential} is to connect equilibrium revenue to \(c_\lambda(P)\), which is a particularly useful statistic because, under the independent prior, its difference from the buyer's expected value per good, \(\mu\), is at most \(\lambda b^2/(8K)\). Accordingly, bounding how far equilibrium revenue per good can lie below \(c_\lambda(P)\) will yield the extraction result.

Let \(D\) be the largest difference between \(c_\lambda(P)\) and equilibrium revenue per good, across all priors and equilibria, and consider an equilibrium for which this difference is nearly \(D\). Write \(c=c_\lambda(P)\) and, to test whether the equilibrium can give the seller so little revenue, consider what she could earn by changing only her current offer. Suppose she offers the grand bundle at the per-good price \(p=c-\delta D-x\) (\(x > 0\)) and, after rejection, follows the strategy prescribed by the equilibrium. 

If the buyer rejects with positive probability, let \(P'\) denote the seller's posterior after observing the rejection. The definition of \(D\) guarantees that the seller's continuation revenue per good is at least \(c_\lambda(P')-D\). Moreover, if the buyer is rejecting with probability \(\varrho\), his incentive to reject, together with the bound on continuation surplus, implies \(p\geq c_\lambda(P')-\delta D)\), and since \(p=c-\delta D-x\), rejection can occur only if \(c_\lambda(P')\leq c-x\). In short, rejection must lower \(c_\lambda\) by at least \(x\). 

But then in \Cref{lem:conditioning}, I prove the easy bound that an event with probability \(\varrho\) can lower \(c_\lambda\) by at most \(\log(1/\varrho)/\lambda\). Hence, \(\varrho\leq\mathrm{e}^{-\lambda x}\). So a larger \(x\) reduces the probability of rejection as well as lowering the price the seller receives upon acceptance. Comparing the revenue from this deviation with equilibrium revenue produces
\[D \leq \frac{x}{1-\delta}+b\mathrm{e}^{-\lambda x} +\frac{\delta}{\mathrm{e}\lambda(1-\delta)};\]
and choosing \(x=(\log t)/\lambda\) yields the proposition's bound.

Recall that \(\mu=\ell+q(b-\ell)\).

\begin{proposition}\label{prop:extraction-bound}
Fix \(0<\ell<b\) and \(q\in(0,1)\). Every equilibrium satisfies
\begin{equation}\label{eq:main-extraction}
0\leq\mu-\frac{R_{K,\delta}}K \leq b\inf_{t\geq1}\left\{ \frac{t}{8K(1-\delta)}+\frac{\log t+2}{t}\right\}.
\end{equation}
\end{proposition}
Under the independent prior, the buyer's value per good for the grand bundle, \(w\), is the average of \(K\) independent values. High values on some goods offset low values on others, so \(w\) becomes increasingly predictable and concentrates around \(\mu\). The statistic \(c_\lambda(P)\) lies below \(\mu\) to the extent that low realizations of \(w\) remain possible. Indeed, \(\mu-c_\lambda(P) = \frac1\lambda\log\E \left[\mathrm e^{-\lambda(w-\mu)}\right]\), and because \(w\) is the average of \(K\) independent bounded values, the logarithm on the right is at most \(\lambda^2b^2/(8K)\). Hence, \(\mu-c_\lambda(P)\leq\frac{\lambda b^2}{8K}\), so this difference shrinks as the number of goods grows.

\Cref{prop:exponential} establishes, in turn, that equilibrium revenue per good cannot lie far below \(c_\lambda(P)\), even though the seller cannot commit to her future offers. Writing \(t=\lambda(1-\delta)b\) and combining the two comparisons,
\[\mu-\frac{R_{K,\delta}}K \leq b\left\{\frac{t}{8K(1-\delta)} + \frac{\log t+2}{t} \right\}.\]
The first term reflects the remaining uncertainty about the buyer's average value for the bundle, whereas the second reflects the buyer's ability to reject the bundle and continue bargaining. Increasing \(t\) reduces the second term but increases the first, so taking the infimum over \(t\geq1\) balances the two. In particular, when \(K(1-\delta)\to\infty\), choosing \(t=\sqrt{K(1-\delta)}\) makes both terms converge to zero. As a result, equilibrium revenue per good converges to \(\mu\).

\begin{corollary}
\eqref{eq:extraction-regime} holds, and expected discounted total surplus per good converges to \(\mu\).
\end{corollary}
\begin{proof}
Suppose that \(N_K\coloneqq K(1-\delta_K)\to\infty\), and take \(t_K=\sqrt{N_K}\) for all sufficiently large \(K\). Evaluating \eqref{eq:main-extraction} at \(t=t_K\) shows that \(\mu-\frac{R_{K,\delta_K}}K\rightarrow 0\). Moreover, \eqref{eq:surplus} and \(U_{K,\delta_K}\geq0\) imply \(0 \leq \frac{U_{K,\delta_K}}K \leq \mu-\frac{R_{K,\delta_K}}K\), so \(U_{K,\delta_K}/K\to0\).\end{proof}

\section{Proving the non-extraction part of \texorpdfstring{\Cref{thm:main} }{theorem}}\label{sec:revision}

In this section, I prove \eqref{eq:revision-regime}. I first derive an upper bound on how much the seller's equilibrium revenue per good can exceed $\ell$, using \Cref{lem:F-bounds,lem:posterior-recursion}. This bound tends to zero as $\delta\uparrow1$ for a fixed number of goods and a fixed positive prior probability that all valuations equal $\ell$ \eqref{eq:general-revision}. Finally, I use the independence of each good's valuation to prove \Cref{prop:revision-bound}, which bounds both $R_{K,\delta}/K-\ell$, the seller's revenue per good above $\ell$, and $(\mu-\ell)-U_{K,\delta}/K$, the difference between $\mu-\ell$ and the buyer's utility per good, by $C(1-\delta)\mathrm e^{AK}$ \eqref{eq:main-revision}. Taking $(1-\delta_K)\mathrm e^{AK}\to0$ then establishes \eqref{eq:revision-regime} and completes \Cref{thm:main}.

Fix an equilibrium and a seller history \(h\), and write \(I=I(h)\), \(m=m(h)\), and \(R=R(h)\). By \eqref{eq:imitation-general}, \(u_v(h)=u_{\hat v}(h)\) whenever \(v_i=\hat v_i\) for every \(i\in I\). Let \(u_L\) denote this common utility when \(v_i=\ell\) for every \(i\in I\), and let \(u_H\) denote it when \(v_i=b\) for every \(i\in I\). Define
\begin{equation}\label{eq:d-and-F}
d\coloneqq m(b-\ell)-(u_H-u_L),\quad \text{\&} \quad F\coloneqq R+\frac{bu_L-\ell u_H}{b-\ell}=(R-m\ell)+u_L+\frac{\ell}{b-\ell}d.
\end{equation}
At a public history \((h,M,s)\), let \(u_L(M,s)\) denote \(u_v(h,M,s)\) when \(v_i=\ell\) for every \(i\in I\), and let \(u_H(M,s)\) denote \(u_v(h,M,s)\) when \(v_i=b\) for every \(i\in I\).\footnote{\Conditionref{iii} of \Cref{def:equilibrium} implies that \(u_v(h,M,s)=u_{\hat v}(h,M,s)\) whenever \(v_i=\hat v_i\) for every \(i\in I\) so that \(u_L(M,s)\) and \(u_H(M,s)\) are well defined.} Write \(R(M,s)\coloneqq R(h,M,s)\), and define
\[d(M,s)\coloneqq m(b-\ell)-\left(u_H(M,s)-u_L(M,s)\right) \quad \text{\&} \quad F(M,s)\coloneqq R(M,s)+\frac{bu_L(M,s)-\ell u_H(M,s)}{b-\ell}.\]
\begin{lemma}\label{lem:F-bounds}
At every seller history,
\begin{equation}\label{eq:F-bounds}
0\leq\frac{\ell}{b-\ell}d\leq F \leq\frac{b}{b-\ell}d\leq mb \qquad \text{\&} \qquad 0\leq R-m\ell\leq F.
\end{equation}
After any \((M,s)\) satisfying \(R(M,s)=R\),
\[0\leq\frac{\ell}{b-\ell}d(M,s)\leq F(M,s) \leq\frac{b}{b-\ell}d(M,s)\leq mb \quad \text{\&} \quad 0\leq R(M,s)-m\ell\leq F(M,s).\]
For \((M,s)\) in support of \(\sigma(\cdot\mid h)\), \(R(M,s)=R\) almost surely and \(F=\E_{M,s}[F(M,s)]\).\end{lemma}

The bounds all follow from the buyer's ability to imitate other types. Since the all-high buyer can imitate the all-low buyer, \(u_H\geq u_L\). Conversely, the all-low buyer can imitate the all-high buyer, losing at most \(m(b-\ell)\) in value. And so \(0\leq d\leq m(b-\ell)\). More importantly, any type can imitate the all-high type. Since his expected discounted value from purchases cannot exceed \(v(I)\), his expected discounted payments cannot exceed \(mb-u_H=m\ell-u_L+d\). Averaging across types then produces \(R-m\ell+u_L\leq d\).

We already have the lower bound for seller revenue from \Cref{lem:baseline}, \(m\ell\), while the all-low buyer can guarantee himself zero utility. Consequently, \(R-m\ell\) and \(u_L\) are nonnegative, and their sum cannot exceed \(d\). \(F\) is simply this sum plus \(\ell d/(b-\ell)\), and so these inequalities establish both bounds on \(F\), as well as \(R-m\ell\leq F\). The bound \(d\leq m(b-\ell)\) also implies \(F\leq mb\), completing \eqref{eq:F-bounds}. The same argument applies after any announcement that preserves the seller's revenue, since the announcement leaves her posterior unchanged. Finally, seller optimality implies \(R(M,s)\leq R\) for every announcement and since the seller's equilibrium announcements yield expected revenue \(R\), it must be an equality almost surely. Averaging the continuation payoffs in the definition of \(F\) then establishes the final equality.

For $\eta\in(0,1]$, let $S_m(\delta,\eta)$ be the supremum of the values $F$ and $F(M,s)$ over all priors, equilibria, seller histories $h$ with $m$ remaining goods and posterior probability at least $\eta$ that all their values equal $\ell$, and announcements $(M,s)$ satisfying $R(h,M,s)=R(h)$. By \Cref{lem:F-bounds},\footnote{To establish the last equality in this expression, suppose the posterior assigns probability one to all remaining values being $\ell$. By \Cref{lem:baseline}, $R\geq m\ell$, while \eqref{eq:surplus} implies $R+u_L\leq m\ell$. Since $u_L\geq0$, we have $R=m\ell$ and $u_L=0$. The expected discounted value of the goods purchased must, therefore, equal $m\ell$. Because $\ell>0$ and $\delta<1$, this requires the buyer to purchase all $m$ goods at the current date with probability one. The all-high type can follow an all-low type's strategy, purchase those goods at the same date, and make the same expected payments, $m\ell$. Thus, $u_H\geq m(b-\ell)$. Combine this with $u_H-u_L\leq m(b-\ell)$ (recalling $u_L=0$) to get equality, so $d=0$ and $F=0$ by \eqref{eq:F-bounds}. The same argument applies after any announcement with $R(M,s)=R=m\ell$, establishing $S_m(\delta,1)=0$.}
\begin{equation}\label{eq:S-initial}
0\leq S_m(\delta,\eta)\leq mb, \qquad S_0(\delta,\eta)=0\qquad  \text{\&} \qquad S_m(\delta,1)=0.
\end{equation}

\begin{lemma}\label{lem:posterior-recursion}
Fix $m\geq1$ and $0<\kappa<\eta_0<\eta_1\leq1$, and suppose $\frac b\ell\frac{\eta_1-\eta_0}{\eta_1-\kappa}<1$. Then
\begin{equation}\label{eq:posterior-recursion}
S_m(\delta,\eta_0) \leq \frac{bm(1-\delta)+\dfrac b\ell \max\left\{S_m(\delta,\eta_1), \max_{j<m}S_j(\delta,\kappa)\right\}}{1-\dfrac b\ell\dfrac{\eta_1-\eta_0}{\eta_1-\kappa}}.
\end{equation}
\end{lemma}

The argument behind this result exploits what the buyer's purchases and rejections reveal about his valuations. Consider a history at which the seller assigns probability between \(\eta_0\) and \(\eta_1\) to all \(m\) remaining goods having value \(\ell\). If some purchase leaves posterior probability at least \(\kappa\) that all these goods had value \(\ell\), an all-low buyer must sometimes make that purchase. The all-high buyer can imitate him, which bounds \(F\) using the continuation problem with fewer goods, together with a term proportional to \(1-\delta\). Similarly, if rejection raises the probability of the all-low buyer to at least \(\eta_1\), the same imitation reasoning bounds \(F\) using the continuation problem with the same number of goods but a higher posterior probability.

Otherwise, every possible purchase reduces the probability of the all-low buyer below \(\kappa\), so rejection must increase it while keeping it below \(\eta_1\). Bayes' rule then bounds the probability of a purchase before the first announcement satisfying either of the preceding conditions by \((\eta_1-\eta_0)/(\eta_1-\kappa)\). At announcements satisfying neither condition, the imitation logic also bounds \(F\) by its discounted continuation after rejection plus \((b/\ell)S_m(\delta,\eta_0)\) times the probability of purchase. Iterating this inequality, the contribution from indefinitely many rejections vanishes, while the total contribution from purchases remains bounded via the preceding probability bound. Combining these observations with the bounds from the first two cases and rearranging establishes \eqref{eq:posterior-recursion}.

\begin{proposition}\label{prop:revision-bound}
Fix $0<\ell<b$ and $q\in(0,1)$. There are finite constants $C,A>0$, depending only on $(\ell,b,q)$, such that every equilibrium satisfies
\begin{equation}\label{eq:main-revision}
0\leq\frac{R_{K,\delta}}K-\ell \leq C(1-\delta)\mathrm{e}^{AK}, \qquad \text{\&} \qquad 0\leq(\mu-\ell)-\frac{U_{K,\delta}}K\leq C(1-\delta)\mathrm{e}^{AK}.
\end{equation}
\end{proposition}

To obtain a bound that remains valid as the number of goods grows, I consider a geometric sequence \(r^n\) of lower bounds on the posterior probability of all-low valuations, with \(r=1-\ell/(4b)\). By choosing \(\eta_0=r^n\), \(\eta_1=r^{n-1}\), and \(\kappa=r^{n+s}\), for a sufficiently large integer \(s\), I can ensure that the denominator in \eqref{eq:posterior-recursion} is at least \(1/2\). The resulting inequality bounds \(S_m(\delta,r^n)\) in terms of a multiple of \(m(1-\delta)\) and continuation bounds with either a higher probability of all-low valuations or fewer remaining goods.  In the first case, observing the buyer's rejection raises the posterior probability of all-low valuations to at least \(r^{n-1}\), reducing \(n\) by one. In the second, the buyer purchases some goods, reducing \(m\) by at least one, although the posterior probability bound may fall to \(r^{n+s}\). In any case \(n+(s+1)m\) decreases.

I now repeatedly apply the bound in \Cref{lem:posterior-recursion} to \(S_m(\delta,r^n)\). Since each application reduces \(n+(s+1)m\) by at least one, we eventually reach either \(n=0\) or \(m=0\). In the first case, the probability of all-low valuations is one, so \(S_m(\delta,1)=0\). In the second, there are no goods left to sell, so \(S_0(\delta,\eta)=0\). Each application of the recursion adds at most \(2bm(1-\delta)\) and multiplies the remaining bound by \(B_0\coloneqq2b/\ell\). Summing the resulting geometric produces \(S_m(\delta,r^n)\leq 2bm(1-\delta)(B_0^{n+(s+1)m}-1)/(B_0-1)\). Accordingly, the bound grows exponentially in \(n+m\).

For an arbitrary probability \(\eta\in(0,1]\), choose \(n=\lceil\log(1/\eta)/\log(1/r)\rceil\), so \(r^n\leq\eta\). Since \(S_K(\delta,\eta)\leq S_K(\delta,r^n)\) and \(B_0^n\leq B_0\eta^{-a}\), where \(a\coloneqq\log B_0/\log(1/r)\), I obtain \(S_K(\delta,\eta)/K\leq [2bB_0/(B_0-1)](1-\delta)B_0^{(s+1)K}\eta^{-a}\). Under the independent prior, the probability that all \(K\) goods have value \(\ell\) is \((1-q)^K\), and substituting in this probability, I obtain an upper bound on \(S_K(\delta,(1-q)^K)/K\) of \([2bB_0/(B_0-1)](1-\delta)\mathrm{e}^{AK}\), where \(A=(s+1)\log B_0+a\log(1/(1-q))\). At the initial history, \Cref{lem:F-bounds} and the definition of \(S_K\) imply \(0\leq R_{K,\delta}/K-\ell\leq F/K\leq S_K(\delta,(1-q)^K)/K\). This establishes the bound on the seller's revenue per good.

To bound the buyer's utility, recall that any buyer type can imitate the all-high type and obtain at least his total valuation minus \(Kb-u_H\).  Averaging across types and recalling that \(u_L\geq0\), we obtain \(U_{K,\delta}\geq K(\mu-\ell)-d\). Conversely, the seller's revenue is at least \(K\ell\), while total expected discounted surplus cannot exceed \(K\mu\), so \(U_{K,\delta}\leq K(\mu-\ell)\). Consequently, \(0\leq(\mu-\ell)-U_{K,\delta}/K\leq d/K\). By \Cref{lem:F-bounds}, \(d\leq(b-\ell)F/\ell\), so the same bound on \(F/K\), multiplied by \((b-\ell)/\ell\), also bounds the difference between the buyer's utility per good and \(\mu-\ell\).

\section{Discussion}

This paper compares two forces that operate in opposite directions. In a static pricing game with many products, aggregating independent values makes a buyer's total willingness to pay predictable, engendering (at least virtual) full extraction of consumer surplus.\footnote{See \citet{BakosBrynjolfsson1999} and \citet{Armstrong1999}. Also related is \citet{JacksonSonnenschein2007}, where linking independent decision problems relaxes incentive constraints.}  But these analyses grant the seller full commitment. On the flip side, there is a sizable literature that studies a durable-good monopolist's inability to commit not to reduce future prices.\footnote{\citet{Coase1972} introduces the eponymous conjecture, and \citet{GulSonnenscheinWilson1986} provide the foundational analysis.} Notably, the Coasian conclusion can fail for a number of distinct reasons.\footnote{\citet{AusubelDeneckere1989} show that, as the time between offers vanishes, equilibria can support every seller payoff between zero and the static monopoly profit. Granularity can also produce non-Coasian outcomes \citep{BagnoliSalantSwierzbinski1989,VonderFehrKuhn1995}. \citet{McAfeeWiseman2008}, \citet{BoardPycia2014}, and \citet{BrzustowskiGeorgiadisHarrisSzentes2023} overturn the Coasian limit through capacity choice, an irreversible outside option, and ``smart'' contracts, respectively.}

Finally, \citet{NavaSchiraldi2019} study a monopolist selling multiple durable varieties to a continuum of unit-demand consumers and show that, in stationary equilibria, profits converge to the optimal static market-clearing level as price revisions become instantaneous. Their consumers choose at most one variety, so products are alternatives, whereas here there is just a single buyer, who may purchase every good, removing goods from future negotiation. \citet[Theorem~2]{Strulovici2017} fixes a multidimensional contract environment and shows that, as negotiation frictions vanish, every equilibrium converges to type-specific efficient contracts--because these contracts generically differ across types, frequent renegotiation need not eliminate the principal's surplus.\footnote{\citet{HartTirole1988} also link renegotiation-proof contracting to Coasian dynamics.} In my model, private information has \(K\) dimensions, but the efficient allocation assigns every good to every type.

\appendix

\section{\texorpdfstring{\secref{sec:extraction}}{sec} Proofs}

\begin{proof}[Proof of \Cref{lem:baseline}]
Fix \(K\) and \(\delta\). Let \(\alpha\) be the minimum of one and the infimum of \(R(h)/(m(h)\ell)\), where the infimum ranges over all priors \(P\in\Delta(V)\), all equilibria under \(P\), and all nonterminal seller histories \(h\).\footnote{The class over which the infimum is taken is nonempty by \Cref{prop:existence}} Obviously, \(0\leq\alpha\leq1\).

For \(\varepsilon\in(0,1-\delta)\), set \(c=\delta\alpha+(1-\delta)-\varepsilon\in(0,1)\). Fix any prior, an equilibrium under that prior, and a nonterminal seller history \(h\), and write \(I=I(h)\), \(m=m(h)\), and \(W=v(I)\). Consider the seller strategy that differs from \(\sigma\) only at \(h\): at \(h\), she announces \(s=0\) and the menu \(M=\{(S,c\ell|S|)\colon S\subseteq I\}\). The menu contains the offer \((I,cm\ell)\). By selecting it, type \(v\) buys all remaining goods and obtains utility \(W-cm\ell\).

For each \(a=(S,c\ell|S|)\in M\) with \(\pi_a>0\), let \(h_a\) be the seller history at the next date, and write \(m'=m(h_a)\), \(W'=v(I(h_a))\), and \(u_v'=u_v(h_a)\). At \(h_a\), the seller's expected revenue satisfies \(R(h_a)\geq\alpha m'\ell\), by the definition of \(\alpha\) when \(m' > 0\) and by \(R(h_a)=0\) when \(m'=0\). Using this, plus plugging \(h_a\) for \(h\) in \eqref{eq:surplus} produces
\begin{equation}\label{eq:baseline-response-utility}
\E_{P_{h_a}}[u_v'] \leq \E_{P_{h_a}}[W']-\alpha m'\ell.
\end{equation}

For every \(v\) with \(\beta_v(a\mid h,M,0)>0\), type \(v\)'s utility from selecting \(a\) is \(W-W'-c\ell(m-m')+\delta u_v'\). Because he can instead select \((I,cm\ell)\), optimality implies
\[W-W'-c\ell(m-m')+\delta u_v' \geq W-cm\ell,\]
or, equivalently, \(W'-c\ell m'\leq\delta u_v'\). Multiplying this inequality by \(P_h(v)\beta_v(a\mid h,M,0)\) and summing over \(v\) and \(a\) delivers
\[\E[W']-c\ell\E[m'] \leq \delta\E[u_v'].\]
By \eqref{eq:bayes}, conditional on the buyer selecting any \(a\in M\) with \(\pi_a>0\), the distribution of \(v\) is \(P_{h_a}\). Using \eqref{eq:baseline-response-utility} for each such \(a\) and averaging over the buyer's selection yields
\[\E[u_v'] \leq \E[W']-\alpha\ell\E[m'].\]
Combining these inequalities yields
\[
(1-\delta)\E[W']
\leq
(c-\delta\alpha)\ell\E[m']
=
(1-\delta-\varepsilon)\ell\E[m'],
\]
and so \(\varepsilon\ell\E[m']\leq0\), since \(W'\geq m'\ell\). Hence, \(m'=0\) almost surely. Consequently, after the seller offers \(M\) at \(h\), the buyer selects \((I,cm\ell)\) almost surely, so the seller receives \(cm\ell\). \Conditionref{ii} of \Cref{def:equilibrium} thus implies \(R(h)\geq cm\ell\). Because the prior, equilibrium, and seller history were arbitrary, the infimum used to define \(\alpha\) is at least \(c\). Since \(c<1\), it follows that \(\alpha\geq c\). Letting \(\varepsilon\downarrow0\) yields \(\alpha\geq\delta\alpha+(1-\delta)\), so \(\alpha=1\).
\end{proof}

Recall that in  \Cref{prop:exponential}, we consider the game with initial inventory $[K]$ and an arbitrary prior $P\in\Delta(V)$, and recall also \(c_\lambda(P)\coloneqq
-\frac1\lambda\log\E_P[\mathrm{e}^{-\lambda w}]\) (with \(w(v)=v([K])/K\)). Before proving the proposition, I note the following lemma:

\begin{lemma}\label{lem:conditioning}
Fix an equilibrium and a seller history \(h\). Suppose that, after the seller chooses \((M,s)\), the buyer selects \(0\) with probability \(\varrho>0\), and let \(P'\coloneqq P_{h,M,s,0}\). Then
\begin{equation}\label{eq:conditioning}
c_\lambda(P')\geq c_\lambda(P_h)+\frac{\log\varrho}{\lambda}.
\end{equation}
\end{lemma}
\begin{proof}
By \eqref{eq:bayes},
\[
P'(v)
= \frac{P_h(v)\beta_v(0\mid h,M,s)}{\varrho}
\leq
\frac{P_h(v)}{\varrho}.
\]
Multiply by \(\mathrm{e}^{-\lambda w(v)}\), sum over \(v\), and take logarithms.
\end{proof}

\begin{proof}[Proof of \Cref{prop:exponential}]
Let \(D\coloneqq \sup\left\{c_\lambda(P)-R/K\right\}\), where the supremum ranges over all \(P\in\Delta(V)\) and all equilibria under \(P\). Note that \(0 \leq D \leq b\).\footnote{By \eqref{eq:uniform-values}, \(R\geq0\), while \(c_\lambda(P)\leq\E_P[w]\leq b\), so \(D\leq b\). Under a degenerate prior, \(c_\lambda(P)=\E_P[w]\), while \eqref{eq:surplus} and \(u_v\geq0\) imply \(R/K\leq\E_P[w]\). Hence, \(D\geq0\).}

Fix \(\varepsilon>0\), and choose a prior \(P\) and an equilibrium under \(P\) such that
\begin{equation}\label{eq:near-supremum}
r\coloneqq R/K=c-D+\xi,
\qquad
c\coloneqq c_\lambda(P),
\qquad \text{\&} \qquad
0\leq\xi\leq\varepsilon.
\end{equation}
For \(x>0\), consider the seller strategy that differs from her equilibrium strategy only at the initial seller history: she chooses \(s=0\) and the menu \(M=\{0,([K],Kp)\}\) with \(p\coloneqq c-\delta D-x\). Let \(\varrho\) denote the probability that the buyer selects \(0\).
\begin{claim}\label{claim:exponential-rejection}
If \(p>0\), then \(\varrho\leq\mathrm{e}^{-\lambda x}\).
\end{claim}
\begin{proof}
Suppose \(p>0\) and \(\varrho>0\). Let \(h'\) be the seller history after the buyer selects \(0\), and write
\[
P'\coloneqq P_{h'},
\qquad
r'\coloneqq R(h')/K,
\qquad \text{\&} \qquad
u'\coloneqq\E_{P'}[u_v(h')]/K.
\]
Rejection leaves all \(K\) goods unsold. Furthermore, the strategies and posteriors prescribed at \(h'\) and at all histories extending \(h'\) form an equilibrium of the \(K\)-good game under prior \(P'\). Hence, by the definition of \(D\) and \eqref{eq:surplus},
\begin{equation}\label{eq:rejection-continuation-bounds}
r'\geq c_\lambda(P')-D
\qquad \text{\&} \qquad
r'+u'\leq\E_{P'}[w].
\end{equation}

For every type \(v\) that selects \(0\) with positive probability, \Conditionref{iii} of \Cref{def:equilibrium} implies \(w(v)-p\leq\delta u_v(h')/K\), so \(\E_{P'}[w]-p\leq\delta u'\), and so
\[
p\geq(1-\delta)\E_{P'}[w]+\delta r'
\geq c_\lambda(P')-\delta D
\geq c+\frac{\log\varrho}{\lambda}-\delta D,
\]
where the second inequality uses \eqref{eq:rejection-continuation-bounds} and \(c_\lambda(P')\leq\E_{P'}[w]\); and the third uses \eqref{eq:conditioning}. Since \(p=c-\delta D-x\), \(\varrho\leq\mathrm{e}^{-\lambda x}\).
\end{proof}

\begin{claim}\label{claim:exponential-shortfall}
\(D\leq b(\log t+2)/t\).
\end{claim}
\begin{proof}
Suppose first that \(p>0\) and \(\varrho>0\). \Conditionref{ii} of \Cref{def:equilibrium} implies \(r\geq(1-\varrho)p+\delta\varrho r'\). Substituting \(r=c-D+\xi\) from \eqref{eq:near-supremum}, \(p=c-\delta D-x\), and \(r'\geq c+\log(\varrho)/\lambda-D\) from \eqref{eq:rejection-continuation-bounds} and \eqref{eq:conditioning} into this inequality, we obtain
\[
(1-\delta)D
\leq
\xi+(1-\varrho)x+(1-\delta)\varrho c
+\frac{\delta}{\lambda}\varrho\log(1/\varrho).
\]
Applying \(c\leq b\), \((1-\varrho)x\leq x\), \(\varrho\leq\mathrm{e}^{-\lambda x}\), and \(\varrho\log(1/\varrho)\leq1/\mathrm{e}\), and then letting \(\varepsilon\downarrow0\), we obtain
\begin{equation}\label{eq:shortfall}
D\leq\frac{x}{1-\delta}+b\mathrm{e}^{-\lambda x}
+\frac{\delta}{\mathrm{e}\lambda(1-\delta)}.
\end{equation}
If \(p>0\) and \(\varrho=0\), \Conditionref{ii} implies \(r\geq p\), so \((1-\delta)D\leq x+\xi\), and hence \eqref{eq:shortfall} after \(\varepsilon\downarrow0\). If \(p\leq0\), \(r\geq0\) and \eqref{eq:near-supremum} imply \(D\leq c+\xi\leq\delta D+x+\xi\), which again implies \eqref{eq:shortfall} after taking \(\varepsilon\downarrow0\).

For \(t>1\), plug in \(x=(\log t)/\lambda\) (\(t=\lambda(1-\delta)b\)) into \eqref{eq:shortfall}:
\[
D
\leq
b\frac{\log t}{t}+\frac bt+\frac{\delta b}{et}
\leq
b\frac{\log t+2}{t}.
\]
For \(t=1\), the same inequality follows by letting \(x\downarrow0\) in \eqref{eq:shortfall}.
\end{proof}

Finally, deduce the proposition from \Cref{claim:exponential-shortfall} and the definition of \(D\).
\end{proof}

\begin{proof}[Proof of \Cref{prop:extraction-bound}]
For \(X\in[0,b]\), let \(\psi(s)\coloneqq \log\E\left[\mathrm{e}^{s(X-\E X)}\right]\). The second derivative \(\psi''(s)\) is the variance of \(X\) under the probability distribution that weights each realization by \(\mathrm{e}^{s(X-\E X)}\). Because \(X\in[0,b]\), this variance is at most \(b^2/4\). Since \(\psi(0)=\psi'(0)=0\), Taylor's theorem implies \(\psi(s)\leq\frac{s^2b^2}{8}\).

Apply this inequality to each \(v_i\) with \(s=-\lambda/K\). Since the \(v_i\)s are independent and have mean \(\mu\),
\begin{equation}\label{eq:mgf}
\log\E\left[\mathrm{e}^{-\lambda w}\right] =
-\lambda\mu
+\sum_{i=1}^K
\log\E\left[
\mathrm{e}^{-(\lambda/K)(v_i-\mu)}
\right] \leq
-\lambda\mu+\frac{\lambda^2b^2}{8K}.
\end{equation}
By \eqref{eq:exponential-statistic} and \eqref{eq:mgf},
\[
c_\lambda(P)
=
-\frac1\lambda
\log\E\left[\mathrm{e}^{-\lambda w}\right]
\geq
\mu-\frac{\lambda b^2}{8K}.
\]

Fix \(t\geq1\) and set \(\lambda=\frac{t}{(1-\delta)b}\), i.e., \(t=\lambda(1-\delta)b\), so \Cref{prop:exponential} applies and implies
\[\frac{R_{K,\delta}}K \geq c_\lambda(P)-b\frac{\log t+2}{t} \geq \mu-b\left\{
\frac{t}{8K(1-\delta)}+\frac{\log t+2}{t}
\right\}.
\]
This holds for every \(t\geq1\), so taking the infimum over \(t\) establishes the upper bound in \eqref{eq:main-extraction}. At the initial seller history, \eqref{eq:surplus} and \(u_v\geq0\) imply \(R_{K,\delta}/K\leq\mu\), which establishes its lower bound.\end{proof}

\section{\texorpdfstring{\secref{sec:revision}}{sec} Proofs}
\begin{proof}[Proof of \Cref{lem:F-bounds}]
In \eqref{eq:imitation-general}, \(u_L-u_H\) has right-hand side \(\sum_{i\in I}(\ell-b)^+=0\), so \(u_L\leq u_H\); and \(u_H-u_L\) has right-hand side \(\sum_{i\in I}(b-\ell)=m(b-\ell)\), so \(u_H-u_L\leq m(b-\ell)\). Hence, \(0\leq d\leq m(b-\ell)\). For any \(v\in V\), plug \(u_H-u_v\) into \eqref{eq:imitation-general} to get \(u_H-u_v\leq\sum_{i\in I}(b-v_i)\); equivalently,
\begin{equation}\label{eq:high-type-imitation}
u_v\geq u_H-\sum_{i\in I}(b-v_i).
\end{equation}
From \(h\) onward, each good in \(I\) can be purchased at most once, so type \(v\)'s expected discounted value from the goods he purchases is at most \(v(I)\). Because his utility equals this value minus his expected discounted payments, those payments are at most \(v(I)-u_v\). By \eqref{eq:high-type-imitation}, \(v(I)-u_v \leq mb-u_H = m\ell-u_L+d\) and averaging over \(v\sim P_h\) implies \(R-m\ell+u_L\leq d\).

By \Cref{lem:baseline}, \(R-m\ell\geq0\), and the buyer can reject every offer, so \(u_L\geq0\). Therefore, \eqref{eq:d-and-F} implies \(F\geq\frac{\ell}{b-\ell}d\) and \(F\geq R-m\ell\). Moreover,
\[
F
=
(R-m\ell+u_L)+\frac{\ell}{b-\ell}d
\leq
d+\frac{\ell}{b-\ell}d
=
\frac{b}{b-\ell}d.
\]
Together with \(0\leq d\leq m(b-\ell)\), these inequalities establish \eqref{eq:F-bounds} at \(h\).

Now fix \((M,s)\) satisfying \(R(M,s)=R\). The all-low and all-high types can again follow one another's strategies from \((h,M,s)\), so \(0\leq d(M,s)\leq m(b-\ell)\). For any type \(v\), following the all-high type's strategy from \((h,M,s)\) implies
\[
u_v(h,M,s)
\geq
u_H(M,s)-\sum_{i\in I}(b-v_i).
\]
Consequently, type \(v\)'s expected discounted payments from \((h,M,s)\) are at most
\[
v(I)-u_v(h,M,s)
\leq
m\ell-u_L(M,s)+d(M,s).
\]
The type distribution at \((h,M,s)\) is \(P_h\). Averaging the last inequality over \(P_h\) implies
\[
R(M,s)-m\ell+u_L(M,s)\leq d(M,s).
\]
Moreover, \(R(M,s)=R\geq m\ell\) by \Cref{lem:baseline}, and the buyer can select \(0\) and reject every later offer, so \(u_L(M,s)\geq0\). Since \(R(M,s)-m\ell\geq0\), \(u_L(M,s)\geq0\), and \(R(M,s)-m\ell+u_L(M,s)\leq d(M,s)\), the definition of \(F(M,s)\) implies
\[
\frac{\ell}{b-\ell}d(M,s)\leq F(M,s)
\leq\frac{b}{b-\ell}d(M,s)
\qquad \text{\&} \qquad
R(M,s)-m\ell\leq F(M,s).
\]
Together with \(0\leq d(M,s)\leq m(b-\ell)\), these are the stated bounds.

For any \((M,s)\), consider the seller strategy that chooses \((M,s)\) at \(h\) and otherwise agrees with \(\sigma\). \Conditionref{ii} of \Cref{def:equilibrium} implies \(R(M,s)\leq R\). On the other hand, the definition of \(R\) implies \(R=\E_{M,s}[R(M,s)]\), where \((M,s)\) is drawn from \(\sigma(\cdot\mid h)\). Hence, \(R(M,s)=R\) almost surely. Similarly, \(u_L=\E_{M,s}[u_L(M,s)]\) and \(u_H=\E_{M,s}[u_H(M,s)]\). Consequently,
\[
\E_{M,s}[F(M,s)]
=
\E_{M,s}[R(M,s)]
+
\frac{
b\E_{M,s}[u_L(M,s)]
-\ell\E_{M,s}[u_H(M,s)]
}{b-\ell}
=
F. \qedhere
\]\end{proof}

\begin{proof}[Proof of \Cref{lem:posterior-recursion}]
Write $S=S_m(\delta,\eta_0)$ and let $B$ denote the numerator in \eqref{eq:posterior-recursion}. Fix an equilibrium and a seller history with inventory $I$ of $m$ goods. Suppose the seller assigns probability $\eta\in[\eta_0,\eta_1)$ to $v_i=\ell$ for every $i\in I$.

After a response that leaves $m'$ goods, let $R'$ be the seller's continuation revenue and let $u_L'$ and $u_H'$ be the continuation utilities of types who value every remaining good at $\ell$ and at $b$, respectively. Measure these payoffs from the next date, and define $d'$ and $F'$ by \eqref{eq:d-and-F}, using $m'$ in place of $m$. If the seller assigns posterior probability at least $\kappa$ to $v_i=\ell$ for every $i\in I$, she also assigns probability at least $\kappa$ to all $m'$ remaining values being $\ell$, since every remaining good belongs to $I$. Thus, $F'\leq S_{m'}(\delta,\kappa)$.

At a seller history $h$ with $I(h)=I$, call an announcement $(M,s)$ satisfying $R(h,M,s)=R(h)$ \textit{controlled} if the buyer chooses some offer $a\neq0$ with positive probability and, after observing that response, the seller assigns probability at least $\kappa$ to $v_i=\ell$ for every $i\in I$, or if the buyer rejects with positive probability and, after observing rejection, the seller assigns probability at least $\eta_1$ to $v_i=\ell$ for every $i\in I$.

\begin{claim}\label{claim:revision-controlled}
Every controlled announcement satisfies $F(M,s)\leq B$.
\end{claim}
\begin{proof}
Suppose first that the buyer chooses an offer $a\neq0$ with positive probability and, after observing this response, the seller assigns probability at least $\kappa$ to $v_i=\ell$ for every $i\in I$. By \eqref{eq:bayes}, at least one type that values these goods at $\ell$ must choose this response with positive probability. All such types value the remaining goods in the same way and can follow one another's strategies, so they have the same continuation utility $u_L(M,s)$. If the offer contains $k$ goods at price $p$, write $m'=m-k$. Choosing this offer and then following his equilibrium strategy gives that type utility $u_L(M,s)$. The all-high type can make the same purchase and follow his own equilibrium strategy thereafter. Thus,
\[
u_L(M,s)=k\ell-p+\delta u_L',
\qquad \text{\&} \qquad
u_H(M,s)\geq kb-p+\delta u_H'.
\]
Substituting these comparisons into the definitions of $d(M,s)$ and $d'$ yields
\begin{equation}\label{eq:deficit-response}
d(M,s)\leq(1-\delta)m'(b-\ell)+\delta d'.
\end{equation}
Every valuation vector with $v_i=\ell$ for every $i\in I$ also assigns value $\ell$ to each of the $m'$ unsold goods. After this response, the seller's posterior probability that all remaining values equal $\ell$ is, therefore, at least $\kappa$. By the definition of $S_{m'}(\delta,\kappa)$, we have $F'\leq S_{m'}(\delta,\kappa)$.

Multiply \eqref{eq:deficit-response} by $b/(b-\ell)$. \Cref{lem:F-bounds} establishes $F(M,s)\leq bd(M,s)/(b-\ell)$ and $d'\leq(b-\ell)F'/\ell$. Consequently,
\[
F(M,s)
\leq bm'(1-\delta)+\frac{\delta b}{\ell}F'
\leq bm'(1-\delta)+\frac{\delta b}{\ell}S_{m'}(\delta,\kappa)
\leq B,
\]
where the last inequality uses $m'<m$, $\delta<1$, and the definition of $B$.

If instead the buyer rejects with positive probability and the seller then assigns probability at least $\eta_1$ to $v_i=\ell$ for every $i\in I$, \eqref{eq:bayes} again implies that some type with these values rejects with positive probability. Rejection involves no payment and leaves all $m$ goods unsold, so the calculation leading to \eqref{eq:deficit-response} applies with $k=p=0$ and $m'=m$. Since $F'\leq S_m(\delta,\eta_1)$, \eqref{eq:deficit-response} and \Cref{lem:F-bounds} imply
\[F(M,s)\leq\frac{b}{b-\ell}d(M,s) \leq bm(1-\delta)+\frac{\delta b}{b-\ell}d' \leq bm(1-\delta)+\frac{\delta b}{\ell}S_m(\delta,\eta_1)
\leq B.\qedhere\]
\end{proof}

For an announcement $(M,s)$, write $\varrho$ for the probability of rejection and $\eta'$ for the seller's posterior probability that $v_i=\ell$ for every $i\in I$ after rejection. In the next two claims, evaluate $R',u_L',u_H',d'$, and $F'$ at the seller history after rejection.

\begin{claim}\label{claim:revision-posterior-increase}
At any announcement $(M,s)$ with $R(M,s)=R$ that is not controlled, $\varrho>0$, $\eta'<\eta_1$, and
\begin{equation}\label{eq:posterior-increase}
\eta\leq\varrho\eta'+(1-\varrho)\kappa,
\qquad \text{\&} \qquad
\eta'\geq\kappa+\frac{\eta-\kappa}{\varrho}\geq\eta.
\end{equation}
At least one type that values all remaining goods at $\ell$ rejects with positive probability.
\end{claim}
\begin{proof}
Since the announcement is not controlled, every purchase response with positive probability leaves the seller with posterior probability below $\kappa$ that $v_i=\ell$ for every $i\in I$. By \eqref{eq:bayes}, averaging the posterior probability of these valuations over the buyer's responses yields $\eta$. If $\varrho=0$, every response is a purchase, so this average is at most $\kappa$, contrary to $\eta\geq\eta_0>\kappa$. Thus, $\varrho>0$.

Rejection has probability $\varrho$, and purchases have total probability $1-\varrho$. Bounding the posterior probability after each purchase by $\kappa$ implies $\eta\leq\varrho\eta'+(1-\varrho)\kappa$. Solving for $\eta'$ and using $\eta>\kappa$ and $\varrho\leq1$ establishes the second part of \eqref{eq:posterior-increase}. Since the announcement is not controlled, $\eta'<\eta_1$. Finally, the joint probability of rejection and $v_i=\ell$ for every $i\in I$ is $\varrho\eta'>0$. Thus, at least one type with these values rejects with positive probability.
\end{proof}

\begin{claim}\label{claim:revision-one-menu}
At any announcement $(M,s)$ with $R(M,s)=R$ that is not controlled,
\begin{equation}\label{eq:one-menu-F}
F(M,s)\leq\delta\varrho F'
+(1-\varrho)\frac b{b-\ell}d(M,s)
\leq\delta\varrho F'+(1-\varrho)\frac b\ell S.
\end{equation}
\end{claim}
\begin{proof}
By \Cref{claim:revision-posterior-increase}, some type that values all remaining goods at $\ell$ rejects with positive probability, so $u_L(M,s)=\delta u_L'$. The all-high type can also reject and follow his equilibrium strategy, so $u_H(M,s)\geq\delta u_H'$. Substituting into the definitions of $F(M,s)$ and $F'$, and using $R(M,s)=R$, yields
\begin{equation}\label{eq:F-rejection}
F(M,s)\leq R-\delta R'+\delta F'.
\end{equation}
If $\varrho=1$, then $R=\delta R'$, so \eqref{eq:F-rejection} reduces to $F(M,s)\leq\delta F'$, which is \eqref{eq:one-menu-F} in this case. Suppose $\varrho<1$, and let $Y$ be the seller's expected discounted revenue from the current date onward, conditional on a purchase at this date. Rejection generates no current payment and continuation revenue $R'$ from the next date, so
\begin{equation}\label{eq:purchase-revenue}
R=\varrho\delta R'+(1-\varrho)Y,
\qquad \text{\&} \qquad
\varrho(R-\delta R')=(1-\varrho)(Y-R).
\end{equation}
Fix a type $v$ and a purchase response that he chooses with positive probability. By \Conditionref{iii} of \Cref{def:equilibrium}, choosing this response and then following his equilibrium strategy gives him utility $u_v(h,M,s)$. His expected discounted value from all goods he purchases from this date onward is at most $v(I)$. His expected discounted payments conditional on this response are thus at most $v(I)-u_v(h,M,s)$. Type $v$ can follow the all-high type's strategy from $(h,M,s)$ onward, obtaining utility at least $u_H(M,s)-\sum_{i\in I}(b-v_i)$. Consequently,
\[v(I)-u_v(h,M,s) \leq v(I)-u_H(M,s)+\sum_{i\in I}(b-v_i) = mb-u_H(M,s).\]
Averaging over types and responses conditional on a purchase, and subtracting $R$, produces \(Y\leq mb-u_H(M,s)\) and
\begin{equation}\label{eq:purchase-payment}
Y-R \leq mb-u_H(M,s)-R
=\frac b{b-\ell}d(M,s)-F(M,s),
\end{equation}
where the equality follows from the definitions of $d(M,s)$ and $F(M,s)$. Multiply \eqref{eq:F-rejection} by $\varrho$. Substituting $\varrho(R-\delta R')=(1-\varrho)(Y-R)$ from \eqref{eq:purchase-revenue}, and then the bound on $Y-R$ in \eqref{eq:purchase-payment}, yields
\[
\varrho F(M,s)
\leq (1-\varrho)(Y-R)+\delta\varrho F'
\leq(1-\varrho)
\left[\frac b{b-\ell}d(M,s)-F(M,s)\right]
+\delta\varrho F'.
\]
Adding $(1-\varrho)F(M,s)$ to both sides establishes the first inequality in \eqref{eq:one-menu-F}. Since $m$ goods remain, $\eta\geq\eta_0$, and $R(M,s)=R$, the definition of $S_m(\delta,\eta_0)$ delivers $F(M,s)\leq S$. Together with $d(M,s)\leq(b-\ell)F(M,s)/\ell$ from \Cref{lem:F-bounds}, this implies $bd(M,s)/(b-\ell)\leq(b/\ell)S$, establishing the second inequality.
\end{proof}

At a seller history with inventory $I$ and posterior probability $\eta\in[\eta_0,\eta_1)$ that $v_i=\ell$ for every $i\in I$, follow the equilibrium strategies and count the current date as $0$. Let $p_*$ be the probability that a purchase occurs before any controlled announcement. For an integer $T\geq0$, let $p_T$ be the probability that such a purchase occurs at one of the dates $0,\ldots,T-1$. Write $p_T(M,s)$ for this probability when the seller makes the current announcement $(M,s)$ and both players follow their equilibrium strategies thereafter. Thus, $p_0=p_0(M,s)=0$, $p_T=\E_{M,s}[p_T(M,s)]$, and $p_T\uparrow p_*$ as $T\to\infty$.

For $T\geq1$, let $p_{T-1}'$ be the probability, at the seller history after rejection, that a purchase occurs within the next $T-1$ dates and before any controlled announcement. Then
\begin{equation}\label{eq:revision-purchase-probabilities}
p_T(M,s)=
\begin{cases}
0, &\text{if $(M,s)$ is controlled},\\
1-\varrho+\varrho p_{T-1}', &\text{otherwise}.
\end{cases}
\end{equation}
The first case follows from the definition of $p_T(M,s)$. In the second, the buyer purchases immediately with probability $1-\varrho$ and rejects with probability $\varrho$.

\begin{claim}\label{claim:revision-purchase-probability}
The probability $p_*$ satisfies \(\eta\leq\kappa p_*+\eta_1(1-p_*)\) and \(p_*\leq\frac{\eta_1-\eta_0}{\eta_1-\kappa}\).
\end{claim}
\begin{proof}
We prove by induction on $T$ that
\[\kappa p_T+\eta_1(1-p_T) \geq \eta \leq\kappa p_T(M,s)+\eta_1(1-p_T(M,s))
\]
at every seller history with inventory $I$ and $\eta\in[\eta_0,\eta_1)$, and after every announcement satisfying $R(M,s)=R$. Both inequalities hold for $T=0$, since the purchase probabilities are zero and $\eta<\eta_1$.

Suppose the inequalities hold for $T-1$. If $(M,s)$ is controlled, then $p_T(M,s)=0$, so the second inequality follows from $\eta<\eta_1$. Otherwise, \Cref{claim:revision-posterior-increase} implies that rejection leaves inventory $I$ and posterior probability $\eta'\in[\eta_0,\eta_1)$. The induction hypothesis, therefore, applies at the seller history after rejection. By \eqref{eq:posterior-increase} and \eqref{eq:revision-purchase-probabilities},
\[
\begin{aligned}
\eta&\leq(1-\varrho)\kappa+\varrho\eta'\\
&\leq(1-\varrho)\kappa
+\varrho\left[\kappa p_{T-1}'+\eta_1(1-p_{T-1}')\right] =\kappa p_T(M,s)+\eta_1(1-p_T(M,s)).
\end{aligned}
\]
Since $R(M,s)=R$ almost surely by \Cref{lem:F-bounds}, averaging over the seller's announcements proves the first inequality and completes the induction. Letting $T\to\infty$ and using $p_T\uparrow p_*$ proves $\eta\leq\kappa p_*+\eta_1(1-p_*)$. Rearranging and using $\eta\geq\eta_0$ produces the claim.
\end{proof}

\begin{claim}\label{claim:revision-stopped-F}
At the fixed seller history, for every integer $T\geq1$, \(F\leq B+\frac b\ell S p_*+\delta^T S\).
\end{claim}
\begin{proof}
We prove by induction on $T$ that
\[F \leq B+\frac b\ell S p_T+\delta^T S\qquad  \text{\&} \qquad F(M,s) \leq B+\frac b\ell S p_T(M,s)+\delta^T S\]
at every seller history with inventory $I$ and $\eta\in[\eta_0,\eta_1)$, and after every announcement satisfying $R(M,s)=R$. For $T=0$, both inequalities follow from $F\leq S$, $F(M,s)\leq S$, and $B\geq0$.

Suppose the inequalities hold for $T-1$. If $(M,s)$ is controlled, \Cref{claim:revision-controlled} implies $F(M,s)\leq B$. Otherwise, rejection leaves inventory $I$ and $\eta'\in[\eta_0,\eta_1)$ by \Cref{claim:revision-posterior-increase}, so the induction hypothesis applies to $F'$. Substituting its bound into \eqref{eq:one-menu-F}, we obtain
\[
\begin{aligned}
F(M,s)&\leq\delta\varrho F'+(1-\varrho)\frac b\ell S\\
&\leq\delta\varrho
\left(B+\frac b\ell S p_{T-1}'+\delta^{T-1}S\right)
+(1-\varrho)\frac b\ell S\\
&\leq B+\frac b\ell S\left(1-\varrho+\varrho p_{T-1}'\right)
+\delta^T S\\
&=B+\frac b\ell S p_T(M,s)+\delta^T S.
\end{aligned}
\]
The third line uses $0<\delta<1$, $0\leq\varrho\leq1$, and the nonnegativity of $B$, $S$, and $p_{T-1}'$; and the last uses \eqref{eq:revision-purchase-probabilities}. Averaging over the seller's announcements and using $F=\E_{M,s}[F(M,s)]$ from \Cref{lem:F-bounds} proves the first inequality and completes the induction. Since $p_T\leq p_*$, the claim follows.
\end{proof}

Let $T\to\infty$ in \(F\leq B+\frac b\ell S p_*+\delta^T S\) and apply \Cref{claim:revision-purchase-probability}. We obtain
\[
F\leq B+\frac b\ell S\frac{\eta_1-\eta_0}{\eta_1-\kappa}.
\]
The induction in \Cref{claim:revision-purchase-probability} also proves $p_T(M,s)\leq(\eta_1-\eta_0)/(\eta_1-\kappa)$ for every $T$ and every announcement with $R(M,s)=R$ at these histories. Substituting this into the bound on $F(M,s)$ in the proof of \Cref{claim:revision-stopped-F} and letting $T\to\infty$ establishes
\[
F(M,s)\leq B+\frac b\ell S\frac{\eta_1-\eta_0}{\eta_1-\kappa}.
\]

At seller histories with $m$ remaining goods and posterior probability at least $\eta_1$ that all their values equal $\ell$, we have $F\leq S_m(\delta,\eta_1)\leq B$. By the definition of $S_m$, the same bound holds for $F(M,s)$ when $R(h,M,s)=R(h)$. These bounds cover every history in the definition of $S$. Taking the supremum gives
\[
S\leq B+\frac b\ell S\frac{\eta_1-\eta_0}{\eta_1-\kappa}.
\]
Subtracting the term involving $S$ on the right and dividing by $1-(b/\ell)(\eta_1-\eta_0)/(\eta_1-\kappa)>0$ proves \eqref{eq:posterior-recursion}.
\end{proof}

\begin{proof}[Proof of \Cref{prop:revision-bound}]
Set \(c\coloneqq b/\ell\), \(r\coloneqq1-\frac1{4c}\), and \(B_0\coloneqq2c\); and choose an integer $s\geq1$ with $r^{s+1}\leq1/2$. For $n\geq1$, set $\eta_0=r^n$, $\eta_1=r^{n-1}$, and $\kappa=r^{n+s}$. Substituting these values into the fraction in the denominator of \eqref{eq:posterior-recursion} and using $c(1-r)=1/4$, we obtain
\[
\frac b\ell\frac{\eta_1-\eta_0}{\eta_1-\kappa}
=\frac{c(1-r)}{1-r^{s+1}}
=\frac{1}{4(1-r^{s+1})}\leq\frac12.
\]
Thus, \Cref{lem:posterior-recursion} applies, and the denominator of the right-hand side of \eqref{eq:posterior-recursion} is at least $1/2$. Since $B_0=2b/\ell$, \eqref{eq:posterior-recursion} implies
\begin{equation}\label{eq:grid-recursion}
S_m(\delta,r^n)
\leq2bm(1-\delta)
+B_0\max\left\{S_m(\delta,r^{n-1}),\max_{j<m}S_j(\delta,r^{n+s})\right\}.
\end{equation}

Define the integer $H(m,n)\coloneqq n+(s+1)m$.

\begin{claim}\label{claim:revision-uniform-S}
For every $m\in\{0,\ldots,K\}$ and integer $n\geq0$,
\begin{equation}\label{eq:uniform-S}
S_m(\delta,r^n)
\leq2bm(1-\delta)
\frac{B_0^{n+(s+1)m}-1}{B_0-1}.
\end{equation}
\end{claim}
\begin{proof}
We prove this by induction on $H(m,n)$. If $m=0$, both sides equal zero by \eqref{eq:S-initial}. If $n=0$, then $r^n=1$ and $S_m(\delta,1)=0$ by the same equation, while the right-hand side is nonnegative.

Now suppose $m,n\geq1$ and the inequality holds for every pair with smaller $H$. The term $S_m(\delta,r^{n-1})$ in \eqref{eq:grid-recursion} has $H(m,n-1)=H(m,n)-1$. For each term $S_j(\delta,r^{n+s})$ with $j<m$,
\[
H(j,n+s)\leq n+s+(s+1)(m-1)=H(m,n)-1.
\]
The induction hypothesis, therefore, applies to every term inside the maximum in \eqref{eq:grid-recursion}. Since $j\leq m$ and $B_0>1$, each term is at most $2bm(1-\delta)(B_0^{H(m,n)-1}-1)/(B_0-1)$. Substituting these bounds into \eqref{eq:grid-recursion} yields
\[
S_m(\delta,r^n)
\leq2bm(1-\delta)
+B_0 2bm(1-\delta)\frac{B_0^{H(m,n)-1}-1}{B_0-1}
=2bm(1-\delta)\frac{B_0^{H(m,n)}-1}{B_0-1},
\]
which establishes \eqref{eq:uniform-S}.
\end{proof}
\eqref{eq:uniform-S} applies at probabilities $r^n$. To bound $S_K(\delta,\eta)$ for any $\eta\in(0,1]$, choose
\[
n=\left\lceil\frac{\log(1/\eta)}{\log(1/r)}\right\rceil,
\qquad \text{\&} \qquad
a\coloneqq\frac{\log B_0}{\log(1/r)}.
\]
Then $r^n\leq\eta$ and $B_0^n\leq B_0\eta^{-a}$. Every history whose posterior assigns probability at least $\eta$ to all remaining values being $\ell$ also assigns probability at least $r^n$ to that event. Thus, $S_K(\delta,\eta)\leq S_K(\delta,r^n)$. Applying \eqref{eq:uniform-S} with $m=K$ to $S_K(\delta,r^n)$, dividing by $K$, and using $B_0^n\leq B_0\eta^{-a}$, we obtain
\begin{equation}\label{eq:general-revision}
\frac{S_K(\delta,\eta)}K
\leq\frac{2bB_0}{B_0-1}
(1-\delta)B_0^{(s+1)K}\eta^{-a}.
\end{equation}
For the independent prior, $\eta=(1-q)^K$. Define
\begin{equation}\label{eq:A-constant}
A\coloneqq(s+1)\log B_0+a\log\frac1{1-q}.
\end{equation}
At the initial history, \eqref{eq:F-bounds} implies $0\leq R_{K,\delta}-K\ell\leq F$, and $F\leq S_K(\delta,(1-q)^K)$ by definition. Substitute $\eta=(1-q)^K$ into \eqref{eq:general-revision}. Since $B_0^{(s+1)K}(1-q)^{-aK}=\mathrm e^{AK}$ by \eqref{eq:A-constant}, we obtain
\[
0\leq\frac{R_{K,\delta}}K-\ell
\leq\frac{F}{K}
\leq\frac{S_K(\delta,(1-q)^K)}K
\leq\frac{2bB_0}{B_0-1}(1-\delta)\mathrm{e}^{AK}.
\]

Finally, averaging \eqref{eq:high-type-imitation} over the prior and using \eqref{eq:d-and-F}, we obtain
\[
U\geq\E_P[v([K])]-Kb+u_H
=\E_P[v([K])]-K\ell+u_L-d
\geq\E_P[v([K])]-K\ell-d.
\]
The opposite inequality $U\leq\E_P[v([K])]-K\ell$ follows from \Cref{lem:baseline} and \eqref{eq:surplus}. Since $d\leq(b-\ell)F/\ell$ by \eqref{eq:F-bounds}, the difference between $\mu-\ell$ and the buyer's equilibrium utility per good under the independent prior satisfies
\[
0\leq(\mu-\ell)-\frac{U_{K,\delta}}K
\leq\frac dK
\leq\frac{b-\ell}{\ell}\frac{S_K(\delta,(1-q)^K)}K
\leq\frac{b-\ell}{\ell}\frac{2bB_0}{B_0-1}(1-\delta)\mathrm{e}^{AK},
\]
where the last inequality uses \eqref{eq:general-revision} and \eqref{eq:A-constant}. Choosing $C=(2bB_0/(B_0-1))\max\{1,(b-\ell)/\ell\}$ proves \eqref{eq:main-revision}. When $(1-\delta_K)\mathrm{e}^{AK}\to0$, both upper bounds in \eqref{eq:main-revision} converge to zero\end{proof}

\bibliography{linkvsdiscount}

\section{(Supplementary) Proof of \texorpdfstring{\Cref{prop:existence}}{prop}}\label[appendix]{sec:existence}
To prove \Cref{prop:existence}, I first obtain equilibria when prices belong to a finite grid and the game has finitely many dates. I then use their continuation payoff vectors to construct equilibrium strategies and beliefs for the infinite-horizon game with arbitrary real prices.

Fix $K,\ell,b,\delta$. The payoff vector at a seller history is $z=(r,(u_v)_{v\in V})$, which includes the utility of every original type, irrespective of whether the current posterior assigns that type positive probability.

For $n\geq1$, let $D_n\coloneqq2^{-n}\mathbb Z\cap[-n,n]$. Call a menu \textit{normalized} if it contains at most one price for each nonempty subset of the remaining inventory (and continues to contain rejection). For an integer $T\geq0$, consider the game in which the seller must choose message $s=0$ and a normalized menu with prices in $D_n$, and play occurs at dates $0,\ldots,T-1$. The game ends when all goods have sold or immediately after the response at date $T-1$, with no subsequent deliveries or payments. If $T=0$ or the initial inventory is empty, it ends immediately. Equilibria of this game must satisfy the conditions in \Cref{def:equilibrium}, with histories and strategies restricted to those of this finite game. Call such a game a \textit{finite grid game}.

\begin{lemma}\label{lem:grid-existence}
For every $n\geq1$, $T\geq0$, initial inventory $I\subseteq[K]$, and prior $P\in\Delta(V)$, the finite grid game has an equilibrium. In every such equilibrium, the continuation payoffs at every seller history $h$ satisfy \(0\leq R(h)\leq |I(h)|b\) and \(0\leq u_v(h)\leq |I(h)|(2b-\ell)\) for every \(v\in V\).
\end{lemma}
\begin{proof}
Fix $n,T,I$. If $T=0$ or $I=\emptyset$, all payoffs are zero. Suppose henceforth that $T\geq1$ and $I\neq\emptyset$.

\begin{claim}\label{claim:existence-full-support}
An equilibrium exists when the prior has full support.
\end{claim}
\begin{proof}
The game has finitely many dates, types, menus, and responses, and it has perfect recall. Thus, \citet[Proposition~1]{KrepsWilson1982} establish that it has a sequential equilibrium. Choose completely mixed strategy profiles\footnote{A strategy profile is completely mixed if the seller assigns strictly positive probability to every available menu whenever she moves and every buyer type assigns strictly positive probability to every response at each buyer history.} converging to its strategies, with Bayesian beliefs converging to its beliefs at seller histories. Write \(\beta_v^k\) for type \(v\)'s strategy in the \(k\)th profile and \(P_g^k\) for the posterior at seller history \(g\) in that profile.

The seller's announcement probability conditional on $h$ is the same for every type. Under each completely mixed profile, Bayes' rule implies
\begin{equation}\label{eq:cl1no1}
P_h^k(v)\beta_v^k(a\mid h,M,0)
=
\left(\sum_{\hat v\in V}P_h^k(\hat v) \beta_{\hat v}^k(a\mid h,M,0)\right)P_{h,M,0,a}^k(v).
\end{equation}
Taking \(k \to \infty\) in \eqref{eq:cl1no1} establishes \Conditionref{iv}. The equilibrium strategies are sequentially rational, and all payoffs are finite, so the resulting assessment is an equilibrium.
\end{proof}

\begin{claim}\label{claim:existence-arbitrary-prior}
An equilibrium exists under every prior.
\end{claim}
\begin{proof}
Fix any $P\in\Delta(V)$ and let \(P_j(v)\coloneqq\left(1-\frac1j\right)P(v)+\frac1{j| V|}\) \((j\geq2)\). Choose an equilibrium under each $P_j$, keeping $n$ and $T$ fixed. The finitely many strategy and posterior probability vectors lie in a finite product of compact simplexes. Take a subsequence along which all these vectors converge. The initial posterior converges to $P$, and \Conditionref{iv} holds in the limit.

At any history, the prescribed payoff and the payoff from any fixed deviation from that history onward are finite sums of products of the relevant action probabilities and, for the seller, the posterior probabilities. So both payoffs converge. Taking limits in the inequality stating that the deviation is not profitable establishes optimality at that history. The limiting assessment is thus an equilibrium under $P$.
\end{proof}

\begin{claim}\label{claim:existence-finite-bounds}
Every equilibrium satisfies the stated payoff bounds.
\end{claim}
\begin{proof}
The continuation payoffs for both players are nonnegative as at every seller history $h$, the seller can guarantee zero by offering only rejection for the remaining dates, and each type can guarantee zero by rejecting throughout. Since each good can be purchased at most once and the discount factors are the same,
\[
R(h)+\sum_{v\in V}P_h(v)u_v(h)\leq |I(h)|b.
\]
Thus, $R(h)\leq |I(h)|b$, and some type $\bar v$ with $P_h(\bar v)>0$ satisfies $u_{\bar v}(h)\leq |I(h)|b$. The imitation argument in \eqref{eq:imitation-general} also applies to the finite game. For every $v\in V$, it implies
\[
u_v(h)\leq u_{\bar v}(h)
+\sum_{i\in I(h)}(v_i-\bar v_i)^+
\leq |I(h)|b+|I(h)|(b-\ell)
=|I(h)|(2b-\ell).\qedhere
\]
\end{proof}

Existence follows from \Cref{claim:existence-arbitrary-prior}, and \Cref{claim:existence-finite-bounds} produces the payoff bounds.\end{proof}

For each $n,T,I$, let $E_{n,T,I}$ be the set of all initial prior-payoff pairs $(P,z)$ of equilibria of the finite grid game. By \Cref{lem:grid-existence}, their payoff vectors lie in the compact box
\[
\mathcal Z_I\coloneqq
[0,|I|b]\times[0,|I|(2b-\ell)]^{| V|}.
\]
Fix an equilibrium of a finite grid game and a seller history. The strategies and beliefs prescribed from that history onward form an equilibrium of the game starting with the remaining inventory, the posterior at that history, and the remaining number of dates. To see this, prescribe the original strategies and beliefs after each new history preceded by the fixed history. The available choices are unchanged, and optimality and updating hold at every subsequent history.

Define
\begin{equation}\label{eq:candidate-C}
\mathcal C_I\coloneqq
\bigcap_{j\geq1}
\overline{\bigcup_{\substack{n\geq j\\T\geq j}}E_{n,T,I}}
\ \subseteq\ \Delta(V)\times\mathcal Z_I,
\end{equation}
where closures are taken in $\Delta(V)\times\mathcal Z_I$, and write $\mathcal C_I(P)\coloneqq\{z \colon (P,z)\in\mathcal C_I\}$.

Fix $(P,z)\in\mathcal C_I$, with $z=(r,(u_v)_v)$ and $I\neq\emptyset$. For a normalized menu $M$, write $(S_a,p_a)$ for offer $a$, with $(S_0,p_0)=(\emptyset,0)$ denoting rejection. $\beta_v(a)$ is the probability that type $v$ chooses offer $a$. After each response $a$, let $P^a$ be the posterior and $z^a=(r^a,(u_v^a)_v)$ the continuation payoff vector. I require $(P^a,z^a)\in\mathcal C_{I\setminus S_a}$ for every $a$.

Write $\pi_a\coloneqq\sum_vP(v)\beta_v(a)$ for the probability of response $a$. I require
\begin{align}
P(v)\beta_v(a)&=\pi_aP^a(v)
&&\text{for every }v,a,\label{eq:spec-bayes}\\
\beta_v(a)>0&\quad\Longrightarrow\quad
v(S_a)-p_a+\delta u_v^a
=\max_c\{v(S_c)-p_c+\delta u_v^c\}
&&\text{for every }v,a.\label{eq:spec-ic}
\end{align}
\eqref{eq:spec-ic} is equivalent to
\begin{equation}\label{eq:spec-ic-closed}
\beta_v(a)\left[
v(S_c)-p_c+\delta u_v^c-v(S_a)+p_a-\delta u_v^a
\right]\leq0
\quad\text{for every }v,a,c.
\end{equation}
If type $v$ chooses offer $a$ with positive probability, he must obtain at least as much utility as he would by rejecting. Since $u_v^0\geq0$, $v(S_a)-p_a+\delta u_v^a\geq\delta u_v^0\geq0$, so
\begin{equation}\label{eq:existence-price-bound}
p_a\leq v(S_a)+\delta u_v^a
\leq |S_a|b+\delta|I\setminus S_a|(2b-\ell)
\leq2Kb.
\end{equation}
In an equilibrium, \Cref{lem:grid-existence} establishes the same bounds on continuation utilities. Thus, \eqref{eq:existence-price-bound} also holds for every offer that some type chooses with positive probability.

Define \(\mathcal R\coloneqq\sum_a\pi_a(p_a+\delta r^a)\).
\begin{lemma}\label{lem:all-deviations}
For every inventory $I$, the set $\mathcal C_I$ is compact and $\mathcal C_I(P)$ is nonempty for every prior $P$. Moreover, $\mathcal C_\emptyset=\Delta(V)\times\{0\}$. For every $(P,z)\in\mathcal C_I$ with $I\neq\emptyset$, every normalized menu $M$, and every message $s\in[0,1]$, there exist a probability distribution $\beta_v$ over the offers for each type $v$ and a posterior $P^a$ and continuation payoff vector $z^a$ after each response $a$ such that $(P^a,z^a)\in\mathcal C_{I\setminus S_a}$ for every $a$, \eqref{eq:spec-bayes} and \eqref{eq:spec-ic} hold, and $\mathcal R\leq r$.
\end{lemma}
\begin{proof}
\begin{claim}\label{claim:existence-candidate-properties}
For every $I$, $\mathcal C_I$ is compact and $\mathcal C_I(P)$ is nonempty for every prior $P$. If $I=\emptyset$, then $\mathcal C_I=\Delta(V)\times\{0\}$.
\end{claim}
\begin{proof}
Each closure in \eqref{eq:candidate-C} is a nonempty compact set, and these sets decrease as $j$ increases. Consequently, their intersection is compact and nonempty. For a fixed $P$, choose an equilibrium with $n=T=k$ for each $k$ by \Cref{lem:grid-existence}. A subsequence of these payoff vectors converges in $\mathcal Z_I$, and the resulting pair belongs to $\mathcal C_I$. Thus, $\mathcal C_I(P)$ is nonempty. By \eqref{eq:candidate-C}, $(P,z)\in\mathcal C_I$ if and only if there are pairs $(P_k,z_k)\in E_{n_k,T_k,I}$ with $n_k\to\infty$, $T_k\to\infty$, and $(P_k,z_k)\to(P,z)$.\footnote{(\(\Longrightarrow\)) choose $n_k,T_k\geq k$ and a pair within distance $1/k$ of $(P,z)$. (\(\Longleftarrow\)) For every fixed $j$, all sufficiently late pairs belong to the union in \eqref{eq:candidate-C}, so their limit belongs to its closure.} If $I=\emptyset$, every payoff is zero, so $\mathcal C_\emptyset=\Delta(V)\times\{0\}$. \end{proof}

Fix $(P,z)$ and $(M,s)$ as in the lemma, and choose equilibria of finite grid games with initial priors and payoff vectors $(P_k,z_k)\in E_{n_k,T_k,I}$ such that $n_k\to\infty$, $T_k\to\infty$, and $(P_k,z_k)\to(P,z)$. Write $z_k=(r_k,(u_{k,v})_v)$. For each offer $(S_a,p_a)\in M$, choose a price $p_{k,a}\in D_{n_k}$ closest to $p_a$ and replace the offer with $(S_a,p_{k,a})$. Call the outcome of this the \textit{rounded menu}. Since $n_k\to\infty$, we have $p_{k,a}\to p_a$ for every $a$. Rejection remains $(\emptyset,0)$ because every grid contains zero.

For a collection \((\beta_{k,v})_{v \in V}\), write $\pi_{k,a}\coloneqq\sum_vP_k(v)\beta_{k,v}(a)$.
\begin{claim}\label{claim:existence-rounded-menu}
For every sufficiently large $k$, there exist a probability distribution $\beta_{k,v}$ for each type $v$, and a posterior $P_k^a$ and continuation payoff vector $z_k^a=(r_k^a,(u_{k,v}^a)_v)$ after each response $a$, such that \((P_k^a,z_k^a)\in E_{n_k,T_k-1,I\setminus S_a}\) for every \(a\); \(P_k(v)\beta_{k,v}(a) =\pi_{k,a}P_k^a(v)\) and \(\beta_{k,v}(a)\left[
v(S_c)-p_{k,c}+\delta u_{k,v}^c -v(S_a)+p_{k,a}-\delta u_{k,v}^a \right] \leq0\), for every $v,a,c$; and \(\sum_a\pi_{k,a}(p_{k,a}+\delta r_k^a)\leq r_k\).
\end{claim}
\begin{proof}
Use the response probabilities, posteriors, and continuation payoff vectors from the chosen equilibrium after the seller announces the rounded menu and message $0$.

From the next seller history onward, the equilibrium strategies and beliefs form an equilibrium of the finite grid game with inventory $I\setminus S_a$, prior $P_k^a$, and $T_k-1$ remaining dates. Thus, $(P_k^a,z_k^a)\in E_{n_k,T_k-1,I\setminus S_a}$.

\Conditionref{iv} implies $P_k(v)\beta_{k,v}(a)=\pi_{k,a}P_k^a(v)$ for every $v,a$. If $\beta_{k,v}(a)>0$, \Conditionref{iii} requires that type $v$ obtain at least as much utility from choosing $a$ as from choosing any other offer $c$ and then following his equilibrium strategy. Consequently,
\[
\beta_{k,v}(a)\left[
v(S_c)-p_{k,c}+\delta u_{k,v}^c
-v(S_a)+p_{k,a}-\delta u_{k,v}^a
\right]\leq0
\quad\text{for every }v,a,c.
\]
When $\beta_{k,v}(a)=0$, the left-hand side is zero.

If the seller announces the rounded menu and follows her equilibrium strategy at every later history, she receives $p_{k,a}$ now and continuation revenue $r_k^a$ from the next date onward when the buyer chooses $a$. Her expected revenue from this deviation is thus $\sum_a\pi_{k,a}(p_{k,a}+\delta r_k^a)$. Her equilibrium revenue is $r_k$, so \Conditionref{ii} implies \(\sum_a\pi_{k,a}(p_{k,a}+\delta r_k^a)\leq r_k\). \end{proof}

\begin{claim}\label{claim:existence-menu-limit}
The sequence $\left((\beta_{k,v})_v,(P_k^a,z_k^a)_a\right)_k$ from \Cref{claim:existence-rounded-menu} has a subsequence converging to $\left((\beta_v)_v,(P^a,z^a)_a\right)$, where $(P^a,z^a)\in\mathcal C_{I\setminus S_a}$ for every $a$, \eqref{eq:spec-bayes} and \eqref{eq:spec-ic} hold, and $\mathcal R\leq r$.
\end{claim}
\begin{proof}
There are finitely many types and offers. Every response probability and posterior probability lies in $[0,1]$, and every continuation payoff vector lies in the compact set $\mathcal Z_{I\setminus S_a}$. Choose a subsequence along which all these probabilities and payoff vectors converge, and denote their limits as in the claim. The limits remain nonnegative and sum to one for each probability distribution.

For every response $a$, \Cref{claim:existence-rounded-menu} establishes $(P_k^a,z_k^a)\in E_{n_k,T_k-1,I\setminus S_a}$. Since $n_k\to\infty$ and $T_k-1\to\infty$, \eqref{eq:candidate-C} implies $(P^a,z^a)\in\mathcal C_{I\setminus S_a}$. This also holds after rejection, which leaves the inventory unchanged.

Taking $k\to\infty$ along this subsequence in $P_k(v)\beta_{k,v}(a)=\pi_{k,a}P_k^a(v)$ and $\sum_a\pi_{k,a}(p_{k,a}+\delta r_k^a)\leq r_k$ from \Cref{claim:existence-rounded-menu} establishes \eqref{eq:spec-bayes} and $\mathcal R\leq r$. Taking the same limit in that claim's inequalities comparing type $v$'s utility from offers $a$ and $c$ establishes \eqref{eq:spec-ic-closed}, equivalently \eqref{eq:spec-ic}. The message $s$ does not appear in these equations or inequalities. Thus, the same response probabilities, posteriors, and payoff vectors satisfy them after $(M,s)$.
\end{proof}

The lemma follows from \Cref{claim:existence-candidate-properties,claim:existence-menu-limit}.
\end{proof}

Set $N\coloneqq| V|$. For fixed \(J \geq 1\), set $s_j\coloneqq j/(J+1)$ for $j=1,\ldots,J$. For normalized menus $M_1,\ldots,M_J$, write $(S_{ja},p_{ja})$ for offer $a$ in $M_j$, with $(S_{j0},p_{j0})=(\emptyset,0)$. Let $\lambda_j$ be the probability that the seller announces $(M_j,s_j)$, and let $\beta_{j,v}$ be type $v$'s probability distribution over the offers after that announcement. After response $a$, let $P^{ja}$ be the posterior and $z^{ja}=(r^{ja},(u_v^{ja})_v)$ the continuation payoff vector. Define \(\pi_{ja}\coloneqq\sum_vP(v)\beta_{j,v}(a)\) and
\(\mathcal R_j\coloneqq\sum_a\pi_{ja}(p_{ja}+\delta r^{ja})\).

\begin{lemma}\label{lem:finite-decomposition}
For every $(P,z)\in\mathcal C_I$ with $I\neq\emptyset$, there exists \(J \leq N+1\) and normalized menus $M_1,\ldots,M_J$ with prices in $[-Kb,2Kb]$, weights, response probabilities, posteriors, and continuation payoff vectors such that $(P^{ja},z^{ja})\in\mathcal C_{I\setminus S_{ja}}$ for every $j,a$, \eqref{eq:spec-bayes} and \eqref{eq:spec-ic} hold for every menu $M_j$, and
\begin{align}
\sum_{j=1}^J\lambda_j&=1,
\qquad \mathcal R_j=r\quad\text{for every }j\qquad  \text{\&}\label{eq:onpath-revenues}\\
u_v&=\sum_{j=1}^J\lambda_j
\max_a\{v(S_{ja})-p_{ja}+\delta u_v^{ja}\}
\quad\text{for every }v.\label{eq:onpath-utilities}
\end{align}
\end{lemma}
\begin{proof}
Choose equilibria of finite grid games with initial priors and payoff vectors $(P_k,z_k)\in E_{n_k,T_k,I}$, where $n_k,T_k\geq k$ and $(P_k,z_k)\to(P,z)$. Write $z_k=(r_k,(u_{k,v})_v)$.

\begin{claim}\label{claim:existence-finite-support}
For every $k$, some lottery over at most $N+1$ of the menus that the seller chooses with positive probability yields $z_k$, without changing the buyer's strategy or the seller's strategy from the next date onward.
\end{claim}
\begin{proof}
Every menu that the seller chooses with positive probability earns $r_k$, since otherwise she could increase her revenue by changing her initial lottery. The vector $(u_{k,v})_v$ is the average of the vectors of utilities that the buyer's types obtain after these menus. These vectors belong to $\R^N$, so Carath\'eodory's theorem permits us to reproduce their average with at most $N+1$ of these menus, where each selected menu still earns $r_k$.
\end{proof}
For each $k$, choose a lottery from \Cref{claim:existence-finite-support}. The number of menus in these lotteries takes only the values $1,\ldots,N+1$, so choose a subsequence for which this number is constant and call it $J$. Use this subsequence for the rest of the proof. Write $M_{k,j}$ for the selected menus and $\lambda_{k,j}$ for their probabilities, $j=1,\ldots,J$.
\begin{claim}\label{claim:existence-bounded-prices}
For every $k$, deleting offers with prices above $2Kb$ leaves all selected menus' prices in $[-Kb,2Kb]$. Keeping the response probabilities, posteriors, and continuation payoff vectors after every remaining response leaves the seller's revenue at $r_k$ after each menu and leaves every type's utility unchanged.
\end{claim}
\begin{proof}
No selected menu contains an offer with price below $-Kb$. If it did, every type could choose that offer and then reject for the remaining dates, obtaining utility greater than $Kb$. Conditional on this menu, expected buyer utility would, therefore, exceed the greatest possible discounted total value, $Kb$. The seller's conditional revenue would be negative, contradicting the fact that the menu earns $r_k\geq0$.

By \eqref{eq:existence-price-bound}, every type assigns probability zero to an offer with price above $2Kb$. Remove these offers from each selected menu. For every remaining response, keep each type's response probability and the posterior and continuation payoff vector after that response. The response probabilities still sum to one for each type. Since these probabilities and posteriors do not change, \Conditionref{iv} still holds after every remaining response. Each type continues to choose only offers that maximize his current utility plus discounted continuation utility--deleting offers removes alternatives without changing the utility from any remaining offer--leaving the seller's revenue and every type's utility unchanged.
\end{proof}

Continue to write $M_{k,j}$ for the menus after these deletions. For offer $a$ in $M_{k,j}$, write $p_{k,ja}$ for its price, $\beta_{k,j,v}(a)$ for type $v$'s response probability, and $P_k^{ja}$ and $z_k^{ja}=(r_k^{ja},(u_{k,v}^{ja})_v)$ for the posterior and continuation payoff vector after that response. Set $\pi_{k,ja}\coloneqq\sum_vP_k(v)\beta_{k,j,v}(a)$.

\begin{claim}\label{claim:existence-lottery-limit}
The sequence of menus, weights, response probabilities, posteriors, and continuation payoff vectors from \Cref{claim:existence-finite-support,claim:existence-bounded-prices} has a convergent subsequence. Its limits satisfy $(P^{ja},z^{ja})\in\mathcal C_{I\setminus S_{ja}}$ for every $j,a$, \eqref{eq:spec-bayes} and \eqref{eq:spec-ic} for every menu $M_j$, and \eqref{eq:onpath-revenues} and \eqref{eq:onpath-utilities}.
\end{claim}
\begin{proof}
There are finitely many possible choices of which subsets each menu offers. Choose a subsequence for which, for each $j$, the menu $M_{k,j}$ offers the same subsets throughout the subsequence. Give offers of the same subset the same label $a$.

All prices lie in $[-Kb,2Kb]$. The weights and probabilities lie in $[0,1]$, and the continuation payoff vectors lie in the compact sets $\mathcal Z_{I\setminus S_{ja}}$. Choose a further subsequence for which all these quantities converge. Write $M_j$, $\lambda_j$, $\beta_{j,v}(a)$, $P^{ja}$, and $z^{ja}$ for the limiting menus, weights, response probabilities, posteriors, and continuation payoff vectors. The limiting probabilities remain nonnegative and sum to one.

For every $j,a$, the finite-game posterior and continuation payoff vector belong to $E_{n_k,T_k-1,I\setminus S_{ja}}$. Since $n_k\to\infty$ and $T_k-1\to\infty$, \eqref{eq:candidate-C} implies $(P^{ja},z^{ja})\in\mathcal C_{I\setminus S_{ja}}$. Taking limits in $P_k(v)\beta_{k,j,v}(a)=\pi_{k,ja}P_k^{ja}(v)$ establishes \eqref{eq:spec-bayes} for every menu $M_j$. Taking limits in the inequalities comparing each type's utility from offers $a$ and $c$, as in \Cref{claim:existence-menu-limit}, establishes \eqref{eq:spec-ic-closed}, and hence \eqref{eq:spec-ic}. Taking limits in $\sum_j\lambda_{k,j}=1$ and $\sum_a\pi_{k,ja}(p_{k,ja}+\delta r_k^{ja})=r_k$ establishes \eqref{eq:onpath-revenues}.

The lottery in \Cref{claim:existence-finite-support} yields utility $u_{k,v}$ for each type $v$, and \Cref{claim:existence-bounded-prices} leaves this utility unchanged. After each menu, type $v$ chooses an offer that maximizes his current utility plus discounted continuation utility. For each $j$, the maximum is over the same finite set of offers throughout the subsequence, and every offer's utility converges. So the maximum converges too. Taking limits in the average of these utilities across menus establishes \eqref{eq:onpath-utilities}.

After announcement $(M_j,s_j)$, let type $v$ choose offer $a$ with probability $\beta_{j,v}(a)$, and assign $P^{ja}$ and $z^{ja}$ after that response. Since the messages $s_j$ differ, each announcement identifies $j$, even when some limiting menus coincide.
\end{proof}

\Cref{claim:existence-lottery-limit} establishes the lemma.
\end{proof}

\begin{lemma}\label{lem:candidate-equilibrium}
For every inventory $I$ and every $(P,z)\in\mathcal C_I$, there is an equilibrium of the infinite-horizon game starting with inventory $I$ and prior $P$ whose initial payoff vector is $z$.
\end{lemma}
\begin{proof}
If $I=\emptyset$, the conclusion is immediate. For each nonempty inventory, choose the menus, probabilities, posteriors, and continuation payoff vectors in \Cref{lem:finite-decomposition} as Borel functions of $(P,z)$. For \Cref{lem:all-deviations}, choose the response probabilities, posteriors, and continuation payoff vectors as Borel functions of $(P,z,M,s)$.\footnote{Fix the number of menus and the subsets of goods in their offers. The conditions in each lemma, with \eqref{eq:spec-ic-closed} in place of \eqref{eq:spec-ic}, define a closed set of parameters and choices, and the choices lie in compact spaces. There are finitely many possibilities for the number of menus and these subsets. Each lemma guarantees a nonempty set of choices for every value of its parameters. The Arsenin-Kunugui selection theorem \citep[Theorem~18.18]{Kechris1995} thus allows us to select from these sets as Borel functions of the parameters.}

We construct strategies and beliefs from pairs $(P_h,z)\in\mathcal C_{I(h)}$, where $P_h$ is the seller's posterior and $z$ is the payoff vector we seek to obtain from $h$ onward. At the initial history, use the given pair $(P,z)$. At a nonterminal seller history $h$, let the seller announce $(M_j,s_j)$ with probability $\lambda_j$, using the menus and probabilities from \Cref{lem:finite-decomposition} for $(P_h,z)$. After $(M_j,s_j)$, let type $v$ choose offer $a$ with probability $\beta_{j,v}(a)$, and use $P^{ja}$ and $z^{ja}$ as the posterior and continuation payoff vector after that response.

After any other announcement $(M,s)$, apply \Cref{lem:all-deviations} to the menu containing the cheapest offer of each subset in $M$ and the message $s$. Use its response probabilities, posteriors, and continuation payoff vectors for those offers. For every more expensive offer, use the same posterior and continuation payoff vector as for the cheapest offer of that subset, and let every type choose it with probability zero. Every type then prefers the cheaper offer, so \eqref{eq:spec-bayes} and \eqref{eq:spec-ic} hold for $M$, and $\mathcal R\leq r$.

After each response $a$, use its pair $(P^a,z^a)$ at the next seller history and repeat this construction whenever goods remain.

\begin{claim}\label{claim:existence-assessment}
The construction defines measurable strategies and beliefs satisfying \Conditionref{iv}.
\end{claim}
\begin{proof}
The construction covers every finite public history. Taking the minimum of finitely many prices and checking whether an announcement equals one of the chosen pairs $(M_j,s_j)$ are Borel operations. Each finite history requires only finitely many applications of these operations and the Borel functions chosen above, so the strategies and beliefs are measurable.

Since the initial posterior is $P$ and \eqref{eq:spec-bayes} holds after every response, \Conditionref{iv} holds.

The seller's payoff comparisons use $P_h$, so we must also check that it equals the actual conditional type distribution when she deviates. Fix a seller history $h_0$ with type distribution $P_{h_0}$. Let the seller use any alternative strategy while each buyer type follows his strategy above. We prove by induction on the date that $P_h$ equals the conditional type distribution at the current history, almost surely. This holds at $h_0$ by the choice of the initial distribution. Suppose it holds at $h$. The seller's announcement distribution depends on $h$, not on the buyer's type, so her announcement leaves this type distribution unchanged. After every announcement, \eqref{eq:spec-bayes} implies
\[
P_h(v)\beta_v(a\mid h,M,s)
=\pi_a(h,M,s)P_{h,M,s,a}(v).
\]
For any Borel set of announcement-response pairs, sum this equality over responses in that set and integrate over the seller's announcement distribution. The left-hand side computes the joint probability of type $v$ and a pair in that set. The right-hand side computes that probability using $P_{h,M,s,a}(v)$ as the posterior. Thus, $P_{h,M,s,a}$ is the conditional type distribution at the next history almost surely. This argument also covers continuous announcement distributions and completes the induction.
\end{proof}

\begin{claim}\label{claim:existence-payoffs}
The assessment is finite, and its continuation payoff vector at every seller history $h$ equals the vector $z$ in the pair $(P_h,z)$.
\end{claim}
\begin{proof}
Fix a seller history $h$ and count the current date as $0$. If the seller follows her strategy above, all prices lie in $[-Kb,2Kb]$, regardless of the buyer's choices. Since at most $K$ purchases occur, for every $T\geq0$ and every type $v$,
\begin{equation}\label{eq:existence-payoff-tails}
\sum_{t\geq T}\delta^t|p_t|\leq2K^2b\delta^T,
\qquad \text{\&} \qquad
\sum_{t\geq T}\delta^t|v(S_t)-p_t|\leq3K^2b\delta^T.
\end{equation}

Now let both players follow the strategies above, with type distribution $P_h$ at $h$. Write $\E_v$ for expectation with the buyer's type fixed at $v$. For the pair $(P_h,z)$ at $h$, \eqref{eq:onpath-revenues} and \eqref{eq:onpath-utilities}, together with \eqref{eq:spec-ic}, imply
\begin{equation}\label{eq:existence-current-payoffs}
r = \sum_{j=1}^J\lambda_j\sum_a\pi_{ja}(p_{ja}+\delta r^{ja})\qquad  \text{\&} \qquad u_v = \sum_{j=1}^J\lambda_j\sum_a\beta_{j,v}(a)
\left(v(S_{ja})-p_{ja}+\delta u_v^{ja}\right).
\end{equation}
Write $r_T$ and $u_{v,T}$ for the revenue and type-$v$ utility in the vector at the seller history reached at date $T$, with both zero after the game ends. Substitute \eqref{eq:existence-current-payoffs} for the continuation values at the next seller history, and repeat for $T$ dates. By \Cref{claim:existence-assessment}, the posterior at each seller history equals the conditional type distribution, so the seller's averages use the probabilities of types at those histories. We obtain \(r =\E\left[\sum_{t=0}^{T-1}\delta^t p_t+\delta^T r_T\right]\) and \(u_v=\E_v\left[ \sum_{t=0}^{T-1}\delta^t\left(v(S_t)-p_t\right)+\delta^T u_{v,T} \right]\). By the definition of $\mathcal Z_I$ and the fact that at most $K$ goods remain at date $T$, we have $0\leq r_T\leq Kb$ and $0\leq u_{v,T}\leq K(2b-\ell)$. These bounds and \eqref{eq:existence-payoff-tails} allow us to take $T\to\infty$, which proves
\begin{equation}\label{eq:existence-continuation-payoffs}
R(h)=r,
\qquad \text{\&} \qquad
u_v(h)=u_v\quad\text{for every }v.
\end{equation}

At any buyer history $(h,M,s)$, the current menu is finite, so its prices are bounded. After the current response, the seller's prices lie in $[-Kb,2Kb]$. Thus, all continuation payoffs are finite at this history too. By \eqref{eq:existence-continuation-payoffs} and \eqref{eq:spec-ic}, \(R(h,M,s) =\sum_a\pi_a(p_a+\delta r^a)\) and
\[u_v(h,M,s) =\sum_a\beta_v(a)
\left(v(S_a)-p_a+\delta u_v^a\right) = \max_a\{v(S_a)-p_a+\delta u_v^a\}.\qedhere\]
\end{proof}

\begin{claim}\label{claim:existence-buyer-optimality}
The buyer is sequentially rational at every history.
\end{claim}
\begin{proof}
Fix a buyer history $(h,M,s)$, a type $v$, and any alternative strategy for that type. In this proof, $\E_v$ takes expectations under this alternative strategy and the seller's strategy above, with the type fixed at $v$. For each response $a$, let $h_a$ be the next seller history. By \Cref{claim:existence-payoffs} and \eqref{eq:spec-ic},
\begin{equation}\label{eq:existence-buyer-response}
v(S_a)-p_a+\delta u_v(h_a)\leq u_v(h,M,s).
\end{equation}
This inequality holds at every buyer history.

Count the current response as date $0$. At a later seller history $h_t$, average \eqref{eq:existence-buyer-response} over the responses under the buyer's alternative strategy and then over the seller's announcements. Averaging $u_v(h_t,M_t,s_t)$ over the seller's announcements yields $u_v(h_t)$, so
\begin{equation}\label{eq:existence-buyer-date}
\E_v\left[v(S_t)-p_t+\delta u_v(h_{t+1})\mid h_t\right]
\leq u_v(h_t)\qquad(t\geq1).
\end{equation}
Average \eqref{eq:existence-buyer-response} over the buyer's choice at date $0$. For dates $1,\ldots,T-1$, multiply \eqref{eq:existence-buyer-date} by $\delta^t$ and take expectations. Adding these inequalities cancels the terms involving $u_v(h_1),\ldots,u_v(h_{T-1})$, leaving
\[
\E_v\left[
\sum_{t=0}^{T-1}\delta^t\left(v(S_t)-p_t\right)
+\delta^T u_v(h_T)
\right]\leq u_v(h,M,s)
\qquad(T\geq1).
\]
Here $h_T$ is the seller history reached at date $T$. The continuation utilities $u_v(h_T)$ are bounded. The seller follows her strategy above, so \eqref{eq:existence-payoff-tails} bounds the absolute utility from dates $T$ onward by $3K^2b\delta^T$. Taking $T\to\infty$ establishes that the deviation is not profitable.
\end{proof}

\begin{claim}\label{claim:existence-seller-optimality}
The seller is sequentially rational at every history.
\end{claim}
\begin{proof}
At every seller history $h$, the response probabilities and continuation payoff vectors satisfy $\mathcal R\leq r$ after every announcement. And so by \Cref{claim:existence-payoffs}, every announcement followed by the strategies above earns the seller at most $R(h)$. By \eqref{eq:existence-price-bound}, every type assigns probability zero to every offer with price above $2Kb$, after every announcement. Accordingly, holding the buyer's strategy fixed, no seller strategy can produce a positive payment above $2Kb$.

Fix a seller history $h_0$ and any seller deviation. In this proof, take expectations under this deviation and the buyer's strategy above, starting with type distribution $P_{h_0}$. Count dates from $h_0$. There are at most $K$ purchases, so
\begin{equation}\label{eq:existence-positive-tail}
\sum_{t\geq T}\delta^t p_t^+\leq2K^2b\delta^T \quad\text{almost surely} \qquad(T\geq0).
\end{equation}
Negative prices require the seller to pay the buyer. The expected discounted total of these payments can be infinite, but \eqref{eq:existence-positive-tail} bounds the total of her discounted positive receipts. Her expected revenue is, therefore, well defined; either finite or $-\infty$. A deviation with revenue $-\infty$ cannot be profitable, so suppose the deviation has finite expected revenue. The expected sum of absolute discounted payments is then finite, because the positive receipts are bounded.

By \Cref{claim:existence-assessment}, $P_{h_t}$ is the conditional type distribution at histories reached under this deviation, almost surely. At $h_t$, average $\mathcal R\leq r$ over the distribution of announcements that the deviating strategy chooses there. By \Cref{claim:existence-payoffs}, the result is
\begin{equation}\label{eq:existence-seller-date}
\E\left[p_t+\delta R(h_{t+1})\mid h_t\right]\leq R(h_t)
\quad\text{almost surely}.
\end{equation}
Subtract $R(h_t)$ from both sides of \eqref{eq:existence-seller-date}, multiply by $\delta^t$, and take expectations. Since $0\leq R(h)\leq Kb$ and the payments have finite expected absolute values, this gives
\[
\E\left[
\delta^t p_t+\delta^{t+1}R(h_{t+1})-\delta^tR(h_t)
\right]\leq0.
\]
Summing these inequalities over $t=0,\ldots,T-1$ cancels the terms involving $R(h_1),\ldots,R(h_{T-1})$, so
\begin{equation}\label{eq:seller-finite-deviation}
\E\left[
\sum_{t=0}^{T-1}\delta^t p_t+\delta^T R(h_T)
\right]\leq R(h_0).
\end{equation}
Since $R(h_T)\geq0$, \eqref{eq:existence-positive-tail} and \eqref{eq:seller-finite-deviation} imply \(\E\left[\sum_{t\geq0}\delta^t p_t\right] \leq R(h_0)+2K^2b\delta^T\), and letting $T\to\infty$ establishes that no seller deviation is profitable.
\end{proof}

The four claims verify \Cref{def:equilibrium}, and \Cref{claim:existence-payoffs} proves that the initial payoff vector is $z$.
\end{proof}

By \Cref{lem:all-deviations}, $\mathcal C_{[K]}(P)$ is nonempty for every $P$. Choosing any vector in this set and applying \Cref{lem:candidate-equilibrium} proves \Cref{prop:existence}.

\end{document}